\documentclass[a4paper,11pt]{article}
\usepackage{float}
\usepackage{amsthm,,amssymb, color, epsfig, graphics}
\usepackage{xcolor}
\usepackage{graphicx}
\usepackage[intlimits]{amsmath}
\usepackage{latexsym}
\usepackage{accents}
\usepackage{epic}

\newtheorem{theorem}{Theorem}
\newtheorem{lemma}{Lemma}

\theoremstyle{definition}

\newtheorem{corollary}{Corollary}
\newtheorem{remark}{Remark}

\def \h#1{\widehat{#1}}
\def \b#1{\overline{#1}}

\def \t#1{\widetilde{#1}}
\def \c#1{\accentset{\circ}{#1}}
\def \th#1{\widehat{\widetilde{#1}}}

\def \dt#1{\underaccent{\tilde}{#1}}
\def \dh#1{\underaccent{\hat}{#1}}
\def \dth#1{\underaccent{\tilde}{\underaccent{\hat}{#1}}}

\usepackage{geometry,multicol}
\newdimen\stockheight
\newdimen\stockwidth
\stockwidth\paperwidth
\numberwithin{equation}{section}

\title{Wronskian solutions of integrable systems}

\author{Da-jun Zhang\footnote{E-mail: djzhang@staff.shu.edu.cn }
   \\
  {\small\it Department of Mathematics, Shanghai University,
    Shanghai 200444, P.R. China}}

\date{\today}

\begin{document}

\maketitle

\begin{abstract}
  \noindent
  Wronski determinant (Wronskian) provides a compact form for $\tau$-functions
  that play roles in a large range of mathematical physics.
  In 1979 Matveev and Satsuma,  independently, obtained solutions in Wronskian form for the Kadomtsev-Petviashvili equation.
  Later, in 1981 these solutions  were constructed from Sato's approach.
  Then in 1983, Freeman and Nimmo  invented the so-called Wronskian technique,
  which allows directly verifying bilinear equations when their solutions are given in terms of Wronskians.
  In this technique the considered bilinear equation is usually reduced to the Pl\"ucker relation on Grassmannians,
  and finding solutions of the bilinear equation is transferred to find a Wronskian vector that is
  defined by a linear differential equation system.
  General solutions of such differential equation systems can be constructed by means of
  triangular Toeplitz matrices.
  In this monograph we review the Wronskian technique and solutions in Wronskian form,
  with supporting instructive examples, including the Korteweg-de Vries (KdV) equation, the modified KdV equation,
  the Ablowitz-Kaup-Newell-Segur hierarchy and reductions, and the lattice potential KdV equation.
  (Dedicated to Jonathan J C Nimmo).

\end{abstract}

\section{Introduction}\label{sec-1}

One of kernel figures in the realm of integrable theory is $\tau$-function,
in terms of which multi-solitons of integrable systems are expressed.
There are several remarkable ways to solve integrable systems and provide $\tau$-functions with explicit forms.
In the Inverse Scattering Transform (IST) $\tau$-functions are written by means of the Cauchy matrices (cf.\cite{AblS-1981}).
Same expressions are also employed in some direct approaches (eg. Cauchy matrix approach \cite{NijAH-JPA-2009,XuZZ-JNMP-2014,ZhaZ-SAPM-2013}
and operator approach \cite{Sch-LAA-2010}).
Hirota's exponential polynomials provide a second form for $\tau$-functions which
can be derived in bilinear method \cite{Hir-book-2004} or constructed using vertex operators \cite{MiwJD-book-2000}.
A third form for $\tau$-functions is the Wronskian form
which was constructed using Darboux transformations \cite{MatS-book-1991} or Wronskian technique \cite{FreN-PLA-1983}.

In 1979 Matveev and Satsuma, independently, derived solutions in Wronskian form for the Kadomtsev-Petviashvili (KP)
equation \cite{Mat-LMP-1979,Sat-JPSJ-1979}.
Two years later  these solutions  were reconstructed from the celebrated Sato approach \cite{Sato-RIMS-1981}.
Then in 1983 Freeman and Nimmo  invented  the Wronskian technique,
which provides a procedure to verify bilinear KP and Korteweg-de Vries (KdV) equations
when their solutions are given in terms of Wronskians \cite{FreN-PLA-1983}.
Soon after it proved popular
in integrable systems \cite{Fre-IMA-1984,FreN-PRSL-1983,Nim-PLA-1983-Toda,NimF-PLA-1983-BSQ,NimF-PLA-1983-RS,NimF-JPA-1984-BT}.
In this technique the considered bilinear equation will be reduced to a known identity,
says, usually, the Pl\"ucker relation on Grassmannians;
and seeking solutions of the bilinear equation is conveyed to find a Wronskian vector that is
defined by a linear differential equation system (LDES for short).
Taking the KdV equation as an example, which has the following LDES
\begin{equation}
\phi_{xx}=A\phi,~~\phi_t=-4\phi_{xxx},
\label{LDES-KdV}
\end{equation}
where $\phi=(\phi_1,\phi_2,\cdots,\phi_N)^T$ and $A$ is arbitrary in $\mathbb{C}_{N\times N}$.
In Freeman-Nimmo's consideration the coefficient matrix $A$ is diagonal with distinct nonzero eigenvalues \cite{FreN-PLA-1983}.
This was generalised to the case of $A$ being a Jordan block \cite{SirHR-PLA-1988} in 1988.
Explicit general solutions of such a LDES with arbitrary $A$ can be written out by means of
either variation of constants method of ordinary differential equations \cite{MaY-TAMS-2005} or
triangular Toeplitz matrices \cite{Zha-2006},
and solutions can then be classified according to canonical forms of $A$.

In this monograph we review the Wronskian technique and solutions in Wronskian form for integrable equations.
Instructive examples include the KdV equation, modified KdV (mKdV) equation,
the Ablowitz-Kaup-Newell-Segur (AKNS) hierarchy and reductions, and lattice potential KdV (lpKdV) equation,
which cover Wronskian, double Wronskian and Casorarian forms of solutions.

The review is organized as follows.
Sec.\ref{sec-2} serves as a preliminary in which we introduce  bilinear equation
and B\"acklund transformation (BT),
notations of Wronskians, some determinantal identities, and triangular Toeplitz matrices.
In Sec.\ref{sec-3}, for the KdV equation we show how the Wronskian technique works
in verifying bilinear equation and BT,
and how to present explicit general solutions of its LDES.
Limit relation between multiple pole and simple pole solutions are also explained.
Sec.\ref{sec-4} is for solutions of the mKdV equation, which exhibits many aspects different from the KdV case.
Sec.\ref{sec-5} serves as a part for double Wronskians and reduction technique,
and Sec.\ref{sec-6} introduces Casoratian technique with the lpKdV equation
as a fully discrete example.
Finally, conclusions are given in Sec.\ref{sec-7}.

\section{Preliminary}\label{sec-2}

\subsection{The KdV stuff}\label{sec-2-1}

Let us go through the stuff of the KdV equation
\begin{equation}
u_t+6uu_x+u_{xxx}=0,
\label{kdv-eq}
\end{equation}
which will serve as a demonstration in Wronskian technique.
It has a Lax pair
\begin{subequations}\label{KdV-Lax}
\begin{align}
& \phi_{xx}+u\phi=-\lambda \phi, \label{KdV-Lax-a}\\
& \phi_{t}=-4\phi_{xxx}-6u\phi_x-3u_x\phi,\label{KdV-Lax-b}
\end{align}
\end{subequations}
where $\lambda$ is a spectral parameter.
Employing the transformation
\begin{equation}
u=2(\ln f)_{xx},
\label{KdV-trans}
\end{equation}
the KdV equation \eqref{kdv-eq} is written as its bilinear form
\begin{equation}
(D_{t}D_{x}+D^4_{x})f\cdot f=0,
\label{KdV-blinear}
\end{equation}
where $D$ is the well-known Hirota's bilinear operator defined by \cite{Hir-PRL-1971,Hir-PTP-1974}
\begin{equation}
D^m_xD^n_y f(x,y)\cdot g(x,y)=(\partial_x-\partial_{x'})^m (\partial_y-\partial_{y'})^n f(x,y)g(x',y')|_{x'=x,y'=y}.
\label{D}
\end{equation}
Hirota gave the following compact form for the $N$-soliton solution of \eqref{KdV-blinear} \cite{Hir-PRL-1971}:
\begin{equation}
f=\sum_{\mu=0,1} \mathrm{exp}\left(\sum^{N}_{j=1} \mu_j \eta_j+\sum^N_{1\leq i<j}\mu_i\mu_j a_{ij}\right),
\label{Nss-f}
\end{equation}
where $\eta_j=k_ix-k_i^3 t +\eta_i^{(0)}$ with $k_i, \eta_i^{(0)}\in \mathbb{R}$, $e^{a_{ij}}=\Bigl(\frac{k_i-k_j}{k_i+k_j}\Bigr)^2$,
and the summation of $\mu$ means to take
all possible $\mu_j=0,1$ $(j=1,2,\cdots, N)$.
A proof of \eqref{Nss-f} satisfying \eqref{KdV-blinear} can be found in \cite{Hir-PRL-1971} and \cite{AblS-1981}.

The bilinear KdV equation \eqref{KdV-blinear} admits a bilinear BT \cite{Hir-PTP-1974}
\begin{subequations}\label{kdv-BT-bil}
\begin{align}
& D^2_xf\cdot g=\lambda fg,\label{kdv-BT-bil-a}\\
& (D_x^3+D_t +3 \lambda D_x)f\cdot g =0, \label{kdv-BT-bil-b}
\end{align}
\end{subequations}
which indicates that if $f$ is a solution of \eqref{KdV-blinear} and we solve the BT \eqref{kdv-BT-bil} to get $g$,
then $g$ will be a solution of \eqref{KdV-blinear} as well and $u=2(\ln g)_{xx}$
provides a second solution to the KdV equation.
Note that taking $\phi=g/f$ together with \eqref{KdV-trans}, the BT \eqref{kdv-BT-bil} will recover the Lax pair \eqref{KdV-Lax},
and vice versa, from \eqref{KdV-Lax} to \eqref{kdv-BT-bil}.

\subsection{Wronskians}\label{sec-2-2}

Wronskian is the determinant of a square matrix where its columns are arranged
with consecutively increasing order derivatives of the first column.
Consider
\begin{equation}
\phi=(\phi_1,\phi_2,\cdots,\phi_N)^T
\label{phi}
\end{equation}
where $\phi_i=\phi_i(x)$ are $C^{\infty}$ functions.
Then a Wronskian with $\phi$ as the first (elementary) column is
$W=|\phi, \phi^{(1)},  \phi^{(2)},\cdots,\phi^{(N-1)}|$,
where $\phi^{(i)}=\partial^i_x\phi$.
It can be more compactly expressed as (cf.\cite{FreN-PLA-1983})
\[W=|0,1,2,\cdots,N-1|=|\h{N-1}|.\]
Due to its special structure, derivatives of a Wronskian have quite simple expressions. For example,
\[W_x = |\h{N-2},N|,~~ W_{xx} = |\h{N-3},N-1,N|+|\h{N-2},N+1|,\]
and if $\phi=\phi(x,y,t)$ with dispersion relation $\phi_y=\phi_{xx}$  and $\phi_t=\phi_{xxx}$,
then one has
\begin{align*}
& W_y = |\h{N-2},N+1|-|\h{N-3},N-1,N|,\\
& W_{t} = |\h{N-2},N+2|-|\h{N-3},N-1,N+1|+|\h{N-4},N-2,N-1,N|.
\end{align*}

A double Wronskian is generated by two elementary column vectors
\begin{equation}
\varphi=(\varphi_1,\varphi_2,\cdots,\varphi_{N+M})^T,~~
\psi=(\psi_1,\psi_2,\cdots,\psi_{N+M})^T,
\label{phipsi-NM}
\end{equation}
with the form
$W=|\varphi, \varphi^{(1)},  \varphi^{(2)},\cdots,  \varphi^{(N-1)};\,
\psi, \psi^{(1)}, \psi^{(2)},\cdots,  \psi^{(M-1)}|$,
and  can be  simply written as (cf.\cite{Nimmo-PLA-1983-NLS})
\[W=|0,1,2,\cdots,N-1;\, 0,1,2,\cdots, M-1|=|\h{N-1};\h{M-1}|.\]
Taking the advantage of its structure, derivatives of a double Wronskian is simple as well, e.g.
$W_x=|\h{N-2},N;\h{M-1}|+|\h{N-1};\h{M-2},M|$.

\subsection{Determinantal identities}\label{sec-2-3}

The Wronskian technique allows directly verifying a solution in Wronskian form of a bilinear equation.
Although a Wronskian provides simple expressions for its derivatives, during the verification
one needs some determinantal identities to simplify high order derivatives.
Finally, the bilinear equation to be verified is reduced to the Pl\"ucker relation (Laplace expansion of a zero-valued determinant).
Let us go through these determinantal identities.

\begin{theorem}\label{thm-2-3-1}\cite{Zha-2006}
Let  $\Xi\in \mathbb{C}_{N\times N}$ and denote its column vectors as $\{\Xi_j\}$;
let $\Omega=(\Omega_{i,j})_{N\times N}$ be an operator matrix (i.e. $\Omega_{i,j}$ are operators),
and denote its column vectors as  $\{\Omega_j\}$.
The following relation holds,
\begin{equation}
\sum^N_{j=1} |\Omega_j * \Xi|
=\sum^N_{j=1}|(\Omega^T)_{j} * \Xi^T|,
\label{id-w-2}
\end{equation}
where
\begin{equation*}
|A_j * \Xi|=|\Xi_1,\cdots,\Xi_{j-1},~A_j \circ\Xi_j,~\Xi_{j+1},\cdots, \Xi_{N}|,
\end{equation*}
and $A_j \circ\Xi_j$ stands for
\begin{equation*}
A_j \circ B_j=(A_{1,j}B_{1,j},~A_{2,j}B_{2,j},\cdots, A_{N,j}B_{N,j})^T,
\end{equation*}
in which $A_j =(A_{1,j},~A_{2,j},\cdots, A_{N,j})^T,~ B_j=(B_{1,j},~B_{2,j},\cdots, B_{N,j})^T$ are $N$th-order vectors.
\end{theorem}

\begin{theorem}\label{thm-2-3-2}

Let $\mathbf{a}_j=(a_{1,j},a_{2,j},\cdots,a_{N,j})^T,~ j=1,\cdots, 2N$
be $N$th-order column vectors over $\mathbb{C}$. The Pl\"ucker relation is described as
\begin{equation}
\sum^{N+1}_{j=1} (-1)^{N+1-j} |\mathbf{a}_1, \mathbf{a}_2, \cdots,\, \mathbf{a}_{j-1},  \mathbf{a}_{j+1},\cdots, \,
\mathbf{a}_{N+1}|\cdot
 |\mathbf{a}_j, \mathbf{a}_{N+2}, \cdots, \,  \mathbf{a}_{2N}|=0.
\label{plu-r-3}
\end{equation}
\end{theorem}

In fact, \eqref{plu-r-3} is a Laplace expansion w.r.t. the first $N$ rows of the following zero-valued determinant
\begin{equation*}
\left|
\begin{array}{cccccc}
a_{1,1} & \cdots & a_{1,N+1} & 0 & \cdots & 0\\
\vdots & \vdots & \vdots & \vdots & \vdots & \vdots \\
a_{N,1} & \cdots & a_{N,N+1} & 0 & \cdots & 0\\
a_{1,1} & \cdots & a_{1,N+1} & a_{1,N+2} & \cdots & a_{1,2N}\\
\vdots & \vdots & \vdots & \vdots & \vdots & \vdots \\
a_{N,1} & \cdots & a_{N,N+1} & a_{N,N+2} & \cdots & a_{N,2N}
\end{array}\right|.
\end{equation*}

Special cases of \eqref{plu-r-3} are the following.
\begin{corollary}\label{cor-2-3-1}
Let $P\in \mathbb{C}_{N\times(N-1)}$, $Q \in \mathbb{C}_{N\times(N-k+1)}$ be the remained of $P$ after removing its arbitrary $k-2$ columns where
$3\leq k < N$, and $\mathbf{a}_i,~ i=1, 2, \cdots, k$, be $N$th-order column vectors. Then one has
\begin{equation}
\sum^{k}_{i=1}(-1)^{i-1} |P, \mathbf{a}_i| \cdot
     |{Q}, \mathbf{a}_1, \cdots, \mathbf{a}_{i-1}, \mathbf{a}_{i+1}, \cdots, \mathbf{a}_{k}|=0,
\quad\quad k \geq 3.
\label{plu-r-2}
\end{equation}
\end{corollary}

This is a practical formula to generate identities used in Wronskian verification. For example,
when $k=4$, \eqref{plu-r-2} yields
\begin{align*}
 &|{P},\mathbf{a}_{1}|
   \cdot|{Q},\mathbf{a}_{2},\mathbf{a}_{3},\mathbf{a}_{4}|
 -|{P},\mathbf{a}_{2}|
   \cdot|{Q},\mathbf{a}_{1},\mathbf{a}_{3},\mathbf{a}_{4}| \\
 &+|{P},\mathbf{a}_{3}|
   \cdot|{Q},\mathbf{a}_{1},\mathbf{a}_{2},\mathbf{a}_{4}|
 -|{P},\mathbf{a}_{4}|
   \cdot|{Q},\mathbf{a}_1,\mathbf{a}_{2},\mathbf{a}_{3}|  = 0,
\end{align*}
and when $k=3$,
\begin{equation}
|M, \mathbf{a},\mathbf{b}||M, \mathbf{c},\mathbf{d}|-|M, \mathbf{a},\mathbf{c}||M, \mathbf{b},\mathbf{d}|
+|M, \mathbf{a},\mathbf{d}||M, \mathbf{b},\mathbf{c}|=0,
\label{plu-r-1}
\end{equation}
where we have taken $P=(Q, \mathbf{p}_{N-1})$, $M=Q$,
$\mathbf{a}=\mathbf{p}_{N-1}, \mathbf{b}=\mathbf{a}_{1}, \mathbf{c}=\mathbf{a}_{2}, \mathbf{d}=\mathbf{a}_{3}$,
where $\mathbf{p}_{N-1}$ is the last column of $P$.

\subsection{Triangular Toeplitz matrices}\label{sec-2-4}

Triangular Toeplitz matrices are used to express general solutions of the LDES like \eqref{LDES-KdV}.
A lower triangular Toeplitz matrix (LTTM) of order $N$ is defined as
\begin{equation}
\mathcal{A}=\left(\begin{array}{cccccc}
a_0 & 0    & 0   & \cdots & 0   & 0 \\
a_1 & a_0  & 0   & \cdots & 0   & 0 \\
\cdots &\cdots &\cdots &\cdots &\cdots &\cdots \\
a_{N-1} & a_{N-2} & a_{N-3}  & \cdots &  a_1   & a_0
\end{array}\right)\in \mathbb{C}_{N\times N}.
\label{A-LTT}
\end{equation}
All such matrices of order $N$ compose a commutative set, denoted by $\widetilde{G}_N$, w.r.t. matrix multiplication.
It is easy to find that $G_N=\bigl \{\mathcal{A} \bigl |~\bigr. \mathcal{A}\in \widetilde{G}_N,~|\mathcal{A}|\neq 0 \bigr\}$
is an Abelian group.

\begin{lemma}
\label{Lem-2-4-1}
\cite{Zha-2006}
1. A $C^{\infty}$ function $\alpha(k)$ can generate a LTTM \eqref{A-LTT} via
\begin{equation}
a_j=\frac{1}{j!}\partial^{j}_{k}\alpha(k),~~j=0,1,\cdots,N-1.
\end{equation}
2. On the other hand, for any $\mathcal{A}$ defined as \eqref{A-LTT},
there exists a complex polynomial $\alpha(z)=\sum^{N-1}_{j=0}\alpha_j z^{N-1-j}$, such that
\begin{equation*}
\partial^{j}_{z} \alpha(z)|_{z=k}=j!a_j,~~j=0,1,\cdots,N-1.
\end{equation*}
3. For any given
$\mathcal{A}\in G_N$, there exist  $\pm \mathcal{B}\in G_N$ such
that $\mathcal{B}^2=\mathcal{A}$.
\end{lemma}

We may also consider a  block LTTM defined as
\begin{equation}
\mathcal{A}^B_{[*]}
=\left(\begin{array}{cccccc}
A_0 & 0    & 0   & \cdots & 0   & 0 \\
A_1 & A_0  & 0   & \cdots & 0   & 0 \\
\cdots &\cdots &\cdots &\cdots &\cdots &\cdots \\
A_{N-1} & A_{N-2} & A_{N-3}  & \cdots &  A_1   & A_0
\end{array}\right)_{2N\times 2N},
\label{A-LTT-block}
\end{equation}
where for $\mathcal{A}^B_{[D]}$ we take
$A_j=\biggl(\begin{array}{cc}
a_{j1} & 0   \\
0 & a_{j2}
\end{array}\biggr)$,
for $\mathcal{A}^B_{[T]}$ we take
$A_j=\biggl(\begin{array}{cc}
a_{j} & 0   \\
b_{j} & a_{j}
\end{array}\biggr)$,
and for $\mathcal{A}^B_{[\epsilon]}$ we take
$A_j=\biggl(\begin{array}{cc}
a_{j} & \epsilon b_{j}   \\
b_j & a_{j}
\end{array}\biggr)$
with
$\epsilon\in \mathbb{R}$ and usually taking $\epsilon = \pm 1$.
Note that with regard to the block LTTM $\mathcal{A}^B_{[*]}$, for each
case of ``*'' with different $A_j$,
there are similar properties as $\mathcal{A}$ holds \cite{Zha-2006}.

\section{The KdV equation}\label{sec-3}

\subsection{Wronskian solution of the bilinear KdV}\label{sec-3-1}

Let us repeat Freeman-Nimmo's Wronskian technique \cite{FreN-PLA-1983} with a slight extension.

\begin{theorem}\label{thm-3-1-1}
The bilinear KdV equation \eqref{KdV-blinear} admits a Wronskian solution
\begin{equation}
f=|\h{N-1}|
\label{f-wr}
\end{equation}
which is composed by the first column vector
$\phi(x,t)=(\phi_1,\phi_2,\cdots,\phi_N)^T$ satisfying LDES
\begin{subequations}\label{wc-1}
\begin{align}
& \phi_{xx}=-A \phi,\label{wc-1-x}\\
& \phi_{t}= -4\phi_{xxx}, \label{wc-1-t}
\end{align}
\end{subequations}
where $A=(a_{ij})_{N\times N}\in \mathbb{C}_{N\times N}$ is arbitrary.
\end{theorem}

\begin{proof}
We calculate derivatives
\begin{align*}
& f_{x}=|\widehat{N-2},N|,~~ f_{xx}=|\widehat{N-3},N-1,N|+|\widehat{N-2},N+1|,\\
& f_{xxx}=|\widehat{N-4},N-2,N-1,N|+2|\widehat{N-3},N-1,N+1|+|\widehat{N-2},N+2|,\\
& f_{xxxx}=|\widehat{N-5},N-3,N-2,N-1,N|+3|\widehat{N-4},N-2,N-1,N+1|\\
&~~~~~~~~~~~~~ +2|\widehat{N-3},N,N+1|+3|\widehat{N-3},N-1,N+2|+|\widehat{N-2},N+3|,\\
& f_{t}=-4(|\widehat{N-4},N-2,N-1,N|-|\widehat{N-3},N-1,N+1|+|\widehat{N-2},N+2|),\\
& f_{tx}=-4(|\widehat{N-5},N-3,N-2,N-1,N|-|\widehat{N-3},N,N+1|\\
& ~~~~~~~~~~~ +|\widehat{N-2},N+3|).
\end{align*}
By  substitution the bilinear KdV equation \eqref{KdV-blinear} yields
\begin{align}
& f_{xt}f-f_xf_t + f_{xxxx}f- 4f_{xxx}f_x + 3(f_{xx})^2 \nonumber\\
=& - 3|\widehat{N-1}|\Bigl(|\widehat{N-5},N-3,N-2,N-1,N|-|\widehat{N-4},N-2,N-1,N+1|\nonumber\\
& -2|\widehat{N-3},N,N+1|-|\widehat{N-3},N-1,N+2|+|\widehat{N-2},N+3|\Bigr)\nonumber\\
& -12|\widehat{N-2},N||\widehat{N-3},N-1,N+1| \nonumber \\
& +3\Bigl(|\widehat{N-3},N-1,N|+|\widehat{N-2},N+1|\Bigr)^{2}.
\label{kdv-bil-2w1}
\end{align}
Next, we make use of Theorem \ref{thm-2-3-1} in which we take $\Omega_{ij}\equiv \partial^2_x$ and $\Xi=(\h{N-1})$.
It follows from \eqref{id-w-2} that
\begin{equation}
\mathrm{Tr}(A)f = |\widehat{N-3},N-1,N|-|\widehat{N-2},N+1|.
\label{id-w-1-ap-1}
\end{equation}
Similarly, we can calculate $\mathrm{Tr}(A)(\mathrm{Tr}(A)f)$.
Then, using equality $f\mathrm{Tr}(A)(\mathrm{Tr}(A)f)=(\mathrm{Tr}(A)f)^2$
we find
\begin{align*}
&|\widehat{N-1}|\Bigl(|\widehat{N-5},N-3,N-2,N-1,N|+2|\widehat{N-3},N,N+1|   \\
& ~~ -|\widehat{N-3},N-1,N+2|-|\widehat{N-4},N-2,N-1,N+1|+|\widehat{N-2},N+3|\Bigr)   \\
=& \Bigl(|\widehat{N-3},N-1,N|-|\widehat{N-2},N+1|\Bigr)^{2},
\end{align*}
which can be used to eliminate some terms in \eqref{kdv-bil-2w1} generated from $f_{xxxx}$ and $f_{xt}$.
As a result, we have
\begin{align*}
& f_{xt}f-f_xf_t + f_{xxxx}f- 4f_{xxx}f_x + 3(f_{xx})^2 \\
=& 12\Bigl(|\widehat{N-1}||\widehat{N-3},N,N+1|-\!|\widehat{N-2},N||\widehat{N-3},N-1,N+1|\\
&~~ +|\widehat{N-3},N-1,N||\widehat{N-2},N+1|\Bigr),
\end{align*}
which is zero if we make use of \eqref{plu-r-1} and take
$M=(\h{N-3}),~ \mathbf{a}=N-2,~ \mathbf{b}=N-1,~ \mathbf{c}=N,~ \mathbf{d}=N+1$.
Thus, a direct verification of Wronskian solution of the bilinear KdV equation is completed.
\end{proof}

\begin{remark}\label{rem-3-1}
Note that within the complex filed $\mathbb{C}$ any square matrix $A$ is similar to a lower triangular canonical form $A'$
by a transform matrix $P$, i.e. $A=P^{-1}A' P$.
Define $\psi=P\phi$. Then \eqref{wc-1} yields $\psi_{xx}=-A'\psi,~\psi_t=-4\psi_{xxx}$
and the Wronskians composed by $\phi$ and $\psi$ are related to each other as $f(\psi)=|P|f(\phi)$.
Obviously, if $f(\phi)$ solves the bilinear KdV equation \eqref{KdV-blinear}, so does $f(\psi)$.
In this sense, we only need to consider the LDES \eqref{wc-1} with a canonically formed $A$.
\end{remark}

\subsection{Wronskian solution of the bilinear BT}\label{sec-3-2}

The bilinear BT \eqref{kdv-BT-bil} can transform solutions of the KdV equation from $N-1$ solitons to $N$ solitons
by taking $g=|\h{N-2}|$ and $f=|\h{N-1}|$.
This was proved by Nimmo and Freeman in \cite{NimF-JPA-1984-BT}.
In the following we give a proof for a more general case.

\begin{theorem}\label{thm-3-2-1}
The Wronskians
\begin{equation}
g=|\h{N-2},\sigma_N|,~~f=|\h{N-1}|
\label{fg-BT-w}
\end{equation}
satisfy the bilinear BT \eqref{kdv-BT-bil} with $\lambda=-a_{NN}$,
where $\phi$ is governed by the LDES \eqref{wc-1} in which $A$ is lower triangular,
and $\sigma_j$ is defined as
\begin{equation}
\sigma_j=(\delta_{j,1},\delta_{j,2},\cdots,\delta_{j,N})^T,~~ \delta_{j,i}=\left\{\begin{array}{ll} 1 & j=i,\\ 0& j\neq i.\end{array}\right.
\label{sigma-j}
\end{equation}
\end{theorem}

\begin{proof}
We only prove \eqref{kdv-BT-bil-a}. Note that
\[
g_{x}=|\widehat{N-3},N-1,\sigma_N|,~~ g_{xx}=|\widehat{N-4},N-2,N-1,\sigma_N|+|\widehat{N-3},N,\sigma_N|.
\]
From  \eqref{kdv-BT-bil-a} we have
\begin{align}
& f_{xx}g -2f_xg_x+fg_{xx}-\lambda  fg \nonumber\\
= & \Bigl(|\widehat{N-3},N-1,N|+|\widehat{N-2},N+1|\Bigr)|\h{N-2},\sigma_N|\nonumber\\
&  -2 |\h{N-2},N||\h{N-3},N-1,\sigma_N|\nonumber\\
&  +|\widehat{N-1}|\Bigl(|\widehat{N-4},N-2,N-1,\sigma_N|+|\widehat{N-3},N,\sigma_N| -\lambda  |\widehat{N-2},\sigma_N|\Bigr).
\label{kdv-BT-w-x2}
\end{align}
To simplify the above form, similar to \eqref{id-w-1-ap-1}, for $g$ we have
\[\mathrm{Tr}(A)g= |\widehat{N-4},N-2,N-1,\sigma_N|-|\widehat{N-3},N,\sigma_N|+\sum_{j=1}^{N}a_{jN}|\h{N-2},\sigma_j|.
\]
Then, making use of the equality generated from $g(\mathrm{Tr}(A)f)=f(\mathrm{Tr}(A)g)$,
\eqref{kdv-BT-w-x2} is reduced to
\begin{align*}
& 2\Bigl(|\widehat{N-3},N-1,N||\widehat{N-2},\sigma_N|-|\widehat{N-2},N||\widehat{N-3},N-1,\sigma_N|\\
&~ +|\widehat{N-1}||\widehat{N-3},N,\sigma_N|\Bigr) - |\widehat{N-1}||\widehat{N-2},\t\sigma_N|,
\end{align*}
where $\t\sigma_N=(a_{1N}, a_{2N}, \cdots, a_{NN}+\lambda)^T$.
The last term is zero if we require $A$ to be lower triangular and $\lambda=-a_{NN}$;
the first three terms together contribute a zero value due to \eqref{plu-r-1}
with $M=(\h{N-3}),~~ \mathbf{a}=N-2,~~ \mathbf{b}=N-1,~~ \mathbf{c}=N,~~ \mathbf{d}=\sigma_N$.

\eqref{kdv-BT-bil-b} can be proved in a similar way. One can refer to \cite{XuanOZ-arxiv-2007} for more details.
\end{proof}

Note that bilinear BT \eqref{kdv-BT-bil} provides a transformation
between two solutions \eqref{fg-BT-w}. As a starting point $g$ must be a solution of the bilinear KdV equation.
The simplest case is $g=1$, i.e. $N=1,~\sigma_1=1$.

\subsection{Classification of solutions}\label{sec-3-3}

Based on Remark \ref{rem-3-1}, it is possible to classify solutions of the KdV equation according to canonical forms of $A$.
In general, from the viewpoint of the IST, under strict complex analysis in direct scattering,
transmission coefficient $T(k)$ admits only distinct simple poles,
which are defined on imaginary axis of $k$-plane.
Each soliton is identified by one of these simple poles.
However, in bilinear approach or Darboux transformation, solutions are derived through more direct ways without restriction on those poles.
For example, multiple-pole solutions can be obtained when $A$ in \eqref{wc-1} is a Jordan block.

In the following, we list out  elementary Wronskian column vectors
in terms of the canonical forms of $A$.
Note that mathematically there is no need to differentiate whether the eigenvalues of $A$ are real or complex.

\vskip 5pt
\noindent
{\it Case 1}
\begin{equation}
A={\rm Diag}(-k^2_1,-k^2_2,\cdots,-k^2_N),
\label{nor-matrix-Case1}
\end{equation}
where  $\{ \lambda_j= -k^2_j\}$ are distinct nonzero numbers.
In this case, $\phi$ satisfying \eqref{wc-1} is given as \eqref{phi}
with
\begin{equation}
\phi_j=b^{+}_{j}e^{\xi_j} + b^{-}_{j}e^{-\xi_j},
\label{entry-Case1-1}
\end{equation}
or equivalently
\begin{equation}
\phi_j=a^{+}_{j}\cosh {\xi_j} + a^{-}_{j}\sinh {\xi_j},
\label{entry-Case1}
\end{equation}
where
\begin{equation}
\xi_j=k_j x -4k^3_j t +\xi^{(0)}_j,
\label{xi}
\end{equation}
and $a^{\pm}_{j}, b^{\pm}_{j}, \xi^{(0)}_j\in \mathbb{C}$.

Note that if all $k_j, a^{\pm}_{j}, b^{\pm}_{j}, \xi^{(0)}_j$ are real, together with
taking $a^{\pm}_{j}=[1\mp(-1)^{j}]/2$ and $0<k_1 <k_2 < \cdots < k_N$ in \eqref{entry-Case1},
the corresponding Wronskian $f$ generates a $N$-soliton
solution. In fact, such a Wronskian can be written as
\begin{equation}
f=
K\cdot \biggl (\prod^{N}_{j=1}e^{-\xi_j}\biggr)
    \sum_{\mu=0, 1}\exp \biggl\{\sum_{j=1}^{N}2\mu_{j}\xi'_{j}+\sum_{1\leq j<l\leq N}
   \mu_{j}\mu_{l}a_{jl} \biggr\},
\label{Hirota-KdV}
\end{equation}
where the sum over $\mu=0,1$ refers to each of $\mu_j=0,1$ for $j=0,1,\cdots, N$, and
\begin{equation*}
\xi'_{j}=\xi_{j}-\frac{1}{4}\sum_{l=1,l\neq j}^{N}a_{jl},
~~e^{a_{jl}}=\biggl(\frac{k_l-k_j}{k_l+k_j}\biggr)^2,~~
K= \prod_{1\leq j<l\leq N}(k_{l}-k_{j}).
\end{equation*}
This is nothing but Hirota's expression \eqref{Nss-f} for $N$-soliton solution.
A similar proof for \eqref{Hirota-KdV} can be found in Ref.\cite{ZhaC-PA-2003}.

If all $k_j$ are pure imaginary, we replace $k_j$ with $i\kappa_j$ where $\kappa_j\in \mathbb{R}$
and $i$ is the imaginary unit,
$\phi_j$ \eqref{entry-Case1} can be rewritten as
\begin{equation}
\phi^{+}_j=a^{+}_{j}\cos {\theta_j} + a^{-}_{j}\sin {\theta_j},~~a^{\pm}_{j}\in \mathbb{C},
\label{entry-Case2}
\end{equation}
where
\begin{equation}
\theta_j=\kappa_j x +4\kappa^3_j t+ \theta^{(0)}_j ,~~\theta^{(0)}_j\in \mathbb{R}.
\label{theta}
\end{equation}

\vskip 5pt
\noindent
{\it Case 2}
\begin{equation}
A=-\t J_N[k_1],~~\t J_N[k_1]\doteq \left(\begin{array}{ccccccc}
k^2_1 & 0    & 0   & \cdots & 0   & 0 & 0 \\
2k_1   & k^2_1  & 0   & \cdots & 0   & 0 & 0\\
1   & 2k_1    & k^2_1 & \cdots & 0   & 0 & 0\\
\cdots &\cdots &\cdots &\cdots &\cdots &\cdots &\cdots \\
0   & 0    & 0   & \cdots & 1 & 2k_1  & k^2_1
\end{array}\right)_{N\times N}.
\label{A-k-KdV}
\end{equation}
A general solution of the LDES \eqref{wc-1} of this case is
\begin{equation}
\phi=\mathcal{A}\mathcal{Q}^{+}+  \mathcal{B}\mathcal{Q}^{-}, ~~ \mathcal{A}, \mathcal{B}\in \t G_N,
\label{gen-sol-case 2}
\end{equation}
where
\begin{equation}
\mathcal{Q}^{\pm}=(\mathcal{Q}^{\pm}_{0},\mathcal{Q}^{\pm}_{1},\cdots,\mathcal{Q}^{\pm}_{N-1})^T,
~~
\mathcal{Q}^{\pm}_{j}=\frac{1}{j!}\partial^{j}_{k_1}\phi^{\pm}_1,
\label{Q-pm-k}
\end{equation}
with $\phi_1^{+}=a^{+}_{1}\cosh {\xi_1},~\phi_1^{-}= a^{-}_{1}\sinh {\xi_1}$,
or $\phi^{\pm}_1=b^{\pm}_{1}e^{\pm\xi_1}$.

One can verify that $\mathcal{Q}^{\pm}$ are two independent solutions of the LDES \eqref{wc-1}.
Since \eqref{wc-1} is linear and LTTMs $\mathcal{A}$ and $\mathcal{B}$ contain enough $2N$ arbitrary constants,
\eqref{gen-sol-case 2} provides a general solution to \eqref{wc-1} when $A$ takes Jordan block \eqref{A-k-KdV}.
Due to the gauge property mentioned in Remark \ref{rem-3-1},
a significant form of \eqref{gen-sol-case 2} is
\begin{equation}
\phi=\mathcal{A}\mathcal{Q}^{+} +  \mathcal{Q}^{-}, ~\mathrm{or} ~~\phi=\mathcal{Q}^{+}+\mathcal{B}\mathcal{Q}^{-}, ~~ \mathcal{A},\mathcal{B}\in G_N.
\label{gen-sol-case 2'}
\end{equation}
Note also that one can take $a^{\pm}_{1}$ or $b^{\pm}_{1}$ to be functions of $k_1$,
but their contribution to $\mathcal{Q}^{\pm}_{j}$ through $\partial^{j}_{k_1}\phi^{\pm}_1$
can be balanced by LTTMs $\mathcal{A}$ and $\mathcal{B}$ in light of Lemma \ref{Lem-2-4-1}.

\vskip 5pt
\noindent
{\it Case 3} ~
\begin{equation}
A=J_{N}[0]
=\left(\begin{array}{cccccc}
0 & 0    & 0   & \cdots & 0   & 0 \\
1   & 0  & 0   & \cdots & 0   & 0 \\
\cdots &\cdots &\cdots &\cdots &\cdots &\cdots \\
0   & 0    & 0   & \cdots & 1   & 0
\end{array}\right)_{N\times N}.
\label{nor-matrix-Case5}
\end{equation}
In this case, we get rational solutions to the KdV equation.
General solution of this case is given by
\begin{equation}
\phi=\mathcal{A}\mathcal{R}^{+}+ \mathcal{B} \mathcal{R}^{-},
~~\mathcal{A},\mathcal{B}\in \widetilde{G}_N£¬
\label{gen-sol-KdV-5}
\end{equation}
where
$\mathcal{R}^{\pm}=(\mathcal{R}^{\pm}_{0},\mathcal{R}^{\pm}_{1},\cdots,\mathcal{R}^{\pm}_{N-1})^T$
with
\begin{equation*}
R^{+}_{j}=\frac{1}{(2j)!}\Bigl [\frac{\partial^{2j}}{{\partial k_1}^{2j}}\cosh \xi_1\Bigr ]_{k_1=0},~~
R^{-}_{j}=\frac{1}{(2j+1)!}\Bigl [\frac{\partial^{2j+1}}{{\partial k_1}^{2j+1}}\sinh \xi_1\Bigr ]_{k_1=0}.
\end{equation*}
Significant form of \eqref{gen-sol-KdV-5} is
\begin{equation}
\phi=\mathcal{A}\mathcal{R}^{+} + \mathcal{R}^{-},~
\mathrm{or} ~~ \phi=\mathcal{R}^{+}+ \mathcal{B}\mathcal{R}^{-},
~~\mathcal{A}, \mathcal{B} \in {G}_N.
\end{equation}

\subsection{Notes}\label{sec-3-4}

In Freeman-Nimmo's proof \cite{FreN-PLA-1983} the coefficient matrix $A$ in the LDES \eqref{wc-1} takes a diagonal form.
This was generalised to a case of $A$ being Jordan form by a trick through taking derivatives for $\phi_1$ w.r.t. $k_1$ \cite{SirHR-PLA-1988}.
With the help of Theorem \ref{thm-2-3-1}, one can implement Wronskian verification starting from the LDES \eqref{wc-1}
with arbitrary $A$, as we have done in Sec.\ref{sec-3-1} for bilinear KdV equation and in Sec.\ref{sec-3-2} for bilinear BT.
Nimmo and Freeman also developed a procedure to get rational solutions in Wronskian form \cite{NimF-PLA-1983-RS}.
Although, as mentioned in \cite{NimF-PLA-1983-RS}, they failed to find more examples than the KdV equation,
their technique is indeed general and valid to a large group of integrable equations (cf.\cite{WuZ-JPA-2003}).
With regard to finding general solutions for the  LDES \eqref{wc-1} with a general $A$,
taking the advantage of the linearity of \eqref{wc-1} and gauge property of its solutions (see Remark \ref{rem-3-1}),
one only needs to consider $A$ being a diagonal or a Jordan block.
To find general solutions of Jordan block case,
\cite{MaY-TAMS-2005} employed variation of constants method of ordinary differential equations,
while here we make use of LTTMs, which is more explicit and convenient.

The name ``positons'' was introduced by Matveev \cite{Mat-PLA-1992-a,Mat-PLA-1992-b} for those  solutions
generated from \eqref{entry-Case2}.
Since when taking $u=0$ \eqref{entry-Case2} is a solution of the Schr\"odinger equation \eqref{KdV-Lax-a} with positive eigenvalue $\lambda=\kappa_j^2$,
the corresponding solutions of the KdV equation bear the name of ``positons''.
Along this line, solitons belong to ``negatons'' \cite{RasSK-JPA-1996},
and when $\lambda\in \mathbb{C}$ the solutions are named as ``complexitons'' \cite{MA-PLA-2002} (see also \cite{Jaw-PLA-1984}).
However, such a classification depends on nonlinear equations (see next section for the mKdV equation),
and it seems that all these solutions of the KdV equation, except solitons, are singular,
because by direct scattering analysis all discrete eigenvalues $\{\lambda_j\}$ have to be simple and negative.

Multiple-pole solutions can be explained as a special limit of simple-pole solutions \cite{MatS-book-1991,Zha-2006}.
To understand this, consider the following scaled Wronskian
\begin{equation}
\frac{W(\phi_1, \phi_2,\cdots,\phi_N)}{\prod^{N}_{j=2}(k_1-k_j)^{j-1}}
\label{limit}
\end{equation}
where $\phi_1=\phi_1(k_1)$ is given as \eqref{entry-Case1} and $\phi_j=\phi_1(k_j)$ for $j=2,3,\cdots,N$.
Implementing Taylor expansion successively for $\phi_j$, $j>1$, at $k_1$, \eqref{limit} turns out to be
$W(\phi_1, \partial_{k_1}\phi_1, \frac{1}{2!}\partial^2_{k_1}\phi_1, \cdots, \frac{1}{(N-1)!}\partial^{N-1}_{k_1}\phi_1)$,
which is a Wronskian for multiple-pole solutions.
Rational solutions can be obtained with a more elaborate procedure by taking $k_1\to 0$ after expansion.

\section{The mKdV equation}\label{sec-4}

\subsection{Wronskian solutions}\label{sec-4-1}

The the mKdV$^+$ (in the following the mKdV for short) equation
\begin{equation}
v_t+6v^2v_x + v_{xxx}=0
\label{mKdV}
\end{equation}
is another typical $(1+1)$-dimensional integrable equation, with a Lax pair
\begin{subequations}\label{mKdV-Lax}
\begin{align}
& \psi_{xx}+2 i v\psi_x=\lambda \psi, \label{mKdV-Lax-a}\\
& \psi_{t}=-4\psi_{xxx}+12 i v\psi_{xx} + (6i v_x-v^2)\psi_x.\label{mKdV-Lax-b}
\end{align}
\end{subequations}
It also has a Lax pair with the AKNS spectral problem
(see \eqref{akns-spectral} with $q=-r=v$).
This equation has solutions in Wronskian form as well, but its LDES contains complex operation (see \eqref{cond-mKdV}) which
makes troubles in obtaining general solutions.

Employing the transformation
\begin{equation}
v=i\biggl(\ln{\frac{f^{*}}{f}}\biggr)_x,
\label{trans-mKdV}
\end{equation}
the mKdV equation \eqref{mKdV} is bilinearized as (cf.\cite{Hir-JPSJ-1972-mKdV})
\begin{subequations}
\begin{eqnarray}
&& (D_{t}+D_{x}^{3}) f^* \cdot f=0,
\label{blinear-mKdV1}\\
&& D_{x}^{2}f^* \cdot f=0, \label{blinear-mKdV2}
\end{eqnarray}
\label{blinear-mKdV}
\end{subequations}
where $*$ stands for the complex conjugate.

\begin{theorem}
\label{thm-4-1-1}
The bilinear mKdV equation \eqref{blinear-mKdV} has a  Wronskian solution
\begin{equation}
f=|\widehat{N-1}|,
\label{wrons-mKdV}
\end{equation}
where the elementary column vector $\phi$ obeys the LDES
\begin{subequations}
\begin{align}
&\phi_{x}= B \phi^*,
\label{cond-a}\\
&\phi_{t}=-4\phi_{xxx},
\label{cond-b}
\end{align}
\label{cond-mKdV}
\end{subequations}
and $|B|\neq 0$ is required.
\end{theorem}

A proof can be found in the Appendix of \cite{ZhaZSZ-RMP-2014}.

\begin{remark}\label{rem-4-1}
In the proof of Theorem \ref{thm-4-1-1}, $f^*$ is written as
\begin{equation}
f^*=|B^*||-1,\h{N-2}|,
\end{equation}
which requires $|B|\neq 0$.
Therefore rational solutions of the mKdV equation cannot be derived from Theorem \ref{thm-4-1-1}.
Rational solutions will be considered separately in Sec.\ref{sec-4-3}.
\end{remark}

\begin{remark}\label{rem-4-2}
There is no gauge property for the LDES \eqref{cond-mKdV}.
Noting that both $\phi$ and $\phi^*$ are involved in  \eqref{cond-a}, for any
$\t B=PBP^{-1}$ which is similar to $B$, the new defined vector $\psi=P\phi$ does not satisfy $\psi_x=\t B \psi^*$.
Without gauge property we cannot construct general solutions  of the LDES \eqref{cond-mKdV}
and conduct a completed classification of solutions according to the canonical forms of $B$.
However, we do solve this problem if we make use of the LDES \eqref{wc-1} of the KdV equation and its gauge property.
See Theorem \ref{thm-4-2-1}.
\end{remark}

\subsection{Classification of solutions}\label{sec-4-2}

\subsubsection{Gauge property: revisit}\label{sec-4-2-1}

Consider \eqref{wc-1-x} in the LDES of the KdV equation and introduce $\mathbb{A}=P^{-1}A P$
which is similar to $A$ with a transform matrix $P$.
Defining $\varphi=P^{-1}\phi$ and $\mathbb{B}=P^{-1}B P^*$, from \eqref{cond-a} we have  $\varphi_{x}= \mathbb{B}\varphi^{*}$.
Obviously Wronskian $f(\phi)=|P|f(\varphi)$,
i.e. $\phi$ and $\varphi$ yield same solution for the mKdV equation through \eqref{trans-mKdV}.
We conclude the above analysis by the following Theorem.

\begin{theorem}
\label{thm-4-2-1}
The  Wronskian \eqref{wrons-mKdV} provides a sloution to the bilinear mKdV equation \eqref{blinear-mKdV}
where the elementary column vector $\varphi$ satisfies
\begin{subequations}
\label{mKdV-condition}
\begin{align}
&\varphi_{xx}=\mathbb{A}\varphi, \label{mKdV-condition-a}\\
&\varphi_{x}=\mathbb{B}\varphi^*, \label{mKdV-condition-b}\\
&\varphi_{t}=-4\varphi_{xxx},\label{mKdV-condition-c}
\end{align}
\end{subequations}
$|\mathbb{B}|\neq 0$ and
\begin{equation}
\mathbb{A}=\mathbb{B}\mathbb{B}^{*}. \label{cc}\\
\end{equation}
\end{theorem}

With this Theorem, and noting the fact that the eigenvalues of $\mathbb{A}$ defined by \eqref{cc} are either real or appear as conjugate pairs
if there are any complex ones \cite{ZhaZSZ-RMP-2014},
one can construct general solutions to the LDES \eqref{mKdV-condition} and implement a full classification
of solutions according to the canonical forms of $\mathbb{A}$ instead of $\mathbb{B}$.

\subsubsection{Solitons}\label{sec-4-2-2}

It can be proved that the case of $\mathbb{A}$ containing $N$ negative eigenvalues yields
only trivial solutions of the mKdV equation \eqref{mKdV} \cite{ZhaZSZ-RMP-2014}.
In the following we consider the case
\begin{equation}
\mathbb{A}={\rm Diag}( \lambda_{1}^2,  \lambda_{2}^2, ~ \cdots, ~\lambda_{N}^2),
\label{mathbb-A-11}
\end{equation}
where, without loss of generality,
we let  $\lambda_j=\varepsilon_j ||k_j||\neq 0$, in which $\varepsilon_j=\pm 1$ and $k_j$ can be either real or complex numbers
with distinct absolute values (modulus).

When $k_{j}\in \mathbb{R}$, i.e. $\lambda_j^2=k_j^2$ in \eqref{mathbb-A-11},
$\mathbb{B}$ takes a form
\begin{equation}
\mathbb{B}={\rm Diag}(k_{1}, k_{2}, ~ \cdots, ~ k_{N}),
\label{BB}
\end{equation}
where $k_j\in \mathbb{R}$.
A solution to the LDES \eqref{mKdV-condition} is $\varphi=(\varphi_{1},  \varphi_{2}, \cdots,  \varphi_{N})^{T}$
where
\begin{equation}
\varphi^{}_{j}= a_{j}^+  e^{\xi_{j}}+ i a_{j}^-
 e^{-\xi_{j}}, ~\xi_{j}=k_{j}x-4k_{j}^{3}t+\xi_{j}^{(0)},~~
 a_{j}^+ , a_{j}^-, k_j, \xi_{j}^{(0)} \in \mathbb{R}.
\label{mkdv-c11}
\end{equation}
Particularly, when $a^+_j=(-1)^{j-1},~a^-_j\equiv 1$,
the Wronskian solution  $f(\varphi)$ can be written as \cite{ZDJ-JPSJ-2002-mkdvscs}
\begin{equation*}
f= K\cdot \biggl (\prod^{N}_{j=1}e^{\xi_j}\biggr)
    \sum_{\mu=0, 1}\exp \biggl\{\sum_{j=1}^{N}\mu_{j}(2\eta_{j}+\frac{\pi}{2}i)+\sum_{1\leq j<l\leq N}
   \mu_{j}\mu_{l}a_{jl} \biggr\},
\end{equation*}
where the sum over $\mu=0,1$ refers to each of $\mu_j=0,1$ for $j=1,2,\cdots, N$, and
\begin{equation*}
\eta_{j}=-\xi_{j}-\frac{1}{4}\sum_{l=2,l\neq j}^{N}a_{jl},
~~e^{a_{jl}}=\biggl(\frac{k_l-k_j}{k_l+k_j}\biggr)^2,
~~K= \prod_{1\leq j<l\leq N}(k_{j}-k_{l}).
\end{equation*}
This coincides with the $N$-soliton solution in Hirota's  form \cite{AblS-1981}.

When $k_j=k_{j1}+ik_{j 2} \in \mathbb{C}$, i.e. $\lambda_{j}^2=k_{j1}^2+k_{j2}^2$ in \eqref{mathbb-A-11},
$\mathbb{B}$ is taken as \eqref{BB} but with $k_j \in \mathbb{C}$,
and the  solution $\varphi$ of  \eqref{mKdV-condition} is composed by
\begin{subequations}
\begin{equation}
 \varphi_j= \gamma_j(a_{j}^+  e^{\xi_{j}}+ i a_{j}^-
e^{-\xi_{j}}), ~\xi_{j}=\lambda_{j}x-4\lambda_{j}^{3}t+\xi_{j}^{(0)}, ~~
a_{j}^+ , a_{j}^-,~ \xi_{j}^{(0)} \in \mathbb{R},
\label{mkdv-c12}
\end{equation}
where
\begin{eqnarray}
\gamma_j=1+\frac{i(\lambda_j-k_{j1})}{k_{j2}},~~
\lambda_j=\varepsilon_j ||k_j||.
\end{eqnarray}
\end{subequations}
However, this yields a same solution to the mKdV equation as \eqref{mkdv-c11} does.
Note that the solution obtained in \cite{NimF-JPA-1984-BT} by Nimmo and Freeman
is the case $k_j=ik_{j 2}$.

With the above discussions we come to the following Remark.

\begin{remark}\label{rem-4-3}
No matter whether the diagonal matrix $\mathbb{B}$  \eqref{BB} is real or complex,
the related Wronskian $f(\varphi)$ ALWAYS generates $N$-soliton solutions to the mKdV equation
when all $\{k_j\}$ have distinct absolute values, and each soliton is identified by $\lambda_j=||k_j||$.
In this sense, unlike the KdV equation, there is no positon-negaton-complexiton classification
for the mKdV equation \eqref{mKdV}.
Define an equivalent relation $\sim$ on the complex plane $\mathbb{C}$ by
\begin{equation}
k_i\sim k_j,~ \mathrm{iff} ~ ||k_i||= ||k_j||.
\label{sim}
\end{equation}
Then the  quotient space $\mathbb{C}/\sim$  denotes the positive half real axis (or positive half imaginary axis),
on which we choose distinct $\{k_j\}$ for solitons.
\end{remark}

\subsubsection{Limit solutions of solitons}\label{sec-4-2-3}

This is the case in which both $\mathbb{A}$ and $\mathbb{B}$ are LTTMs.
To find a general solution of this case, we consider $\mathbb{A}=\t J_N[k_1]$ defined  as \eqref{A-k-KdV} with $k_1\in \mathbb{R}$.
Thus we can make use of \eqref{gen-sol-case 2} and write
\begin{equation}
\varphi=\mathcal{A}\mathcal{Q}^{+}+ \mathcal{B} \mathcal{Q}^{-},~~~
\mathcal{A}, \mathcal{B} \in \widetilde{G}_N,
\label{gen-sol-L-1}
\end{equation}
which solves \eqref{mKdV-condition-a} and \eqref{mKdV-condition-c},
where
\begin{equation}
\label{mathbb-Q}
 \mathcal{Q}^{\pm}=(\mathcal{Q}^{\pm}_{0}, \mathcal{Q}^{\pm}_{1}, \cdots, \mathcal{Q}^{\pm}_{N-1})^T,~~
\mathcal{Q}^{\pm}_{s}=\frac{1}{s!}\partial^{s}_{k_1}e^{\pm \xi_1},
\end{equation}
and $\xi_1$ is defined in \eqref{mkdv-c11}.
The matrix $\mathbb{B}$ that satisfies \eqref{cc} is a standard Jordan block
\begin{equation}
\mathbb{B}=J_N[k_1]\doteq\left(\begin{array}{cccccc}
k_1 & 0    & 0   & \cdots & 0   & 0 \\
1   & k_1  & 0   & \cdots & 0   & 0 \\
\cdots &\cdots &\cdots &\cdots &\cdots &\cdots \\
0   & 0    & 0   & \cdots & 1   & k_1
\end{array}\right)_N.
\label{mathbb-B-21}
\end{equation}
In the following we impose extra conditions on $\mathcal{A}$ and $\mathcal{B}$
so that \eqref{gen-sol-L-1} solves \eqref{mKdV-condition-b} as well.
To do that, substituting \eqref{gen-sol-L-1} into  \eqref{mKdV-condition-b}, one has
\begin{equation}
\mathcal{A} \mathcal{Q}^{+}_{x}+\mathcal{B} \mathcal{Q}^{-}_{x}
=\mathbb{B}({\mathcal{A}}^* \mathcal{Q}^{+}+{\mathcal{B}}^*\mathcal{Q}^{-}).
\label{gen-sol-relate}
\end{equation}
Meanwhile, it can be verified that
$\mathcal{Q}^{\pm}_{0, x}=\pm \mathbb{B} \mathcal{Q}_{0}^{\pm}$,
Then, noting that $\mathbb{B}$ is real and $\mathbb{B}$, $\mathcal{A}$ and $\mathcal{B}$ are commutative,
it follows that
\begin{equation}
\mathbb{B}(\mathcal{A} \mathcal{Q}^{+}-\mathcal{B}\mathcal{Q}^{-})
=\mathbb{B}({\mathcal{A}}^* \mathcal{Q}^{+}+{\mathcal{B}}^*\mathcal{Q}^{-}).
\label{gen-sol-relate-new}
\end{equation}
Since $\mathcal{Q}^{+}$ and $\mathcal{Q}^{-}$ are linearly independent, this indicates that
$\mathcal{A}$ is real and $\mathcal{B}$ is pure imaginary.
Thus, a general solution to the LDES \eqref{mKdV-condition} of the mKdV equation is
\begin{equation}
{\varphi}=\mathcal{A^+} \mathcal{Q}^{+}+ i \mathcal{A^-}\mathcal{Q}^{-},
\label{gen-sol-21}
\end{equation}
where $\mathcal{A^\pm}$ are real LTTMs.

The case of $k_1\in \mathbb{C}$ contributes same solutions.
For more details of analysis of this case, one can refer to \cite{ZhaZSZ-RMP-2014}.
Note that we use subtitle ``limit solutions of solitons'' for this part
because the Wronskian $f$ composed by \eqref{gen-sol-21} can also be obtained by taking limits from solitons.
Of course, this is a case in which multiple-pole solutions are related to solitons.

\subsubsection{Breathers and limit breathers}\label{sec-4-2-4}

Breathers are obtained when $\mathbb{A}$ has $N$ distinct complex conjugate-pairs of eigenvalues, with a canonical form
\begin{equation}
\mathbb{A}={\rm Diag}( k_{1}^2,  k^{*2}_{1}, ~ \cdots, ~k_{N}^2, k^{*2}_{N}),~~k_{j}=k_{1j}+ik_{2j},~ k_{1j} k_{2j}\neq 0,
\label{mKdV-matrix-C4}
\end{equation}
and $\mathbb{B}$ takes a special block diagonal form
\begin{equation}
\mathbb{B}={\rm Diag}(\Theta_{1},  \Theta_{2}, \cdots, \Theta_{N}), ~~~ \Theta_{j}=\left(\begin{array}{cc}
0 & k_{j}  \\
k^*_{j} & 0  \\
\end{array}\right).
\label{B-breather}
\end{equation}
Solution to the LDES \eqref{mKdV-condition} is
\begin{subequations}
\label{breather}
\begin{equation}
{\varphi}=(\varphi_{11}, \varphi_{12}, \varphi_{21}, \varphi_{22},
\cdots, \varphi_{N1}, \varphi_{N2})^{T},
\end{equation}
where
\begin{align}
& \varphi_{j1}= a_{j}  e^{\xi_j}+b_{j} e^{-{\xi}_j}, ~~
\varphi_{j2}=a^*_{j}  e^{{\xi}^*_j} -b^*_{j}
e^{-{{\xi}^*_j}}, \\
& \xi_j=k_{j}x- 4 k_j^3 t+ \xi_{j}^{(0)},~~a_{j}, b_{j},\xi_{j}^{(0)}  \in \mathbb{C}.
\label{xi-breather}
\end{align}
\end{subequations}

Limit solutions of breathers are obtained when
 \begin{equation}
\mathbb{A}=\left(
\begin{array}{ccccccc}
  \mathcal {K} & 0 & 0 & \ldots & 0 & 0 & 0\\
 \widetilde{\mathcal {K}} &\mathcal {K} & 0 & \ldots & 0 & 0 & 0\\
 I_2& \widetilde{\mathcal {K}} &\mathcal {K}  & \ldots & 0 & 0 & 0\\
   \ldots &  \ldots &  \ldots &  \ldots &  \ldots &  \ldots&  \ldots\\
 0 & 0 & 0 & \ldots & I_2& \widetilde{\mathcal {K}} &\mathcal {K}
\end{array}\right)_{2N\times 2N},
\end{equation}
and
\begin{equation}
 \mathbb{B}=\left(
\begin{array}{ccccc}
  \Theta_{1} & 0 & \ldots & 0 & 0 \\
  \widetilde{I}_2 & \Theta_{1} & \ldots & 0 & 0 \\
   \ldots &  \ldots &  \ldots &  \ldots &  \ldots\\
  0 &  0 & \ldots &  \widetilde{I}_{2} & \Theta_{1}
\end{array}\right)_{2N\times 2N},
\end{equation}
where
\begin{equation*}
\mathcal {K} = \left(
\begin{array}{cc}
k^2_1&  0 \\
 0 & k_1^{*2}
\end{array}\right), ~~
\widetilde{\mathcal {K}} = \left(
\begin{array}{cc}
2k_1 &  0  \\
 0 & 2{k}^*_1
\end{array}\right),  ~~I_2=\left(
\begin{array}{cc}
1 &  0  \\
   0 & 1
\end{array}\right),~~
\widetilde{I}_{2} = \left(
\begin{array}{cc}
0 &  1  \\
 1 & 0
\end{array}\right).
\end{equation*}
In this case, a general solution of \eqref{mKdV-condition} is given through rewriting
$\varphi=({\varphi^{+}}^T, {\varphi^{-}}^T)^T$ where
$\varphi^{\pm}=(\varphi^{\pm}_1,\varphi^{\pm}_2,\cdots,\varphi^{\pm}_N)^T$ satisfy
\begin{subequations}
\label{entry-condition}
\begin{align}
&  \varphi _{xx}^ +  = \mathbb{A}'{\varphi ^ + }, ~~ \varphi _{xx}^ - =\mathbb{A}'^{*}  {\varphi ^ - },\label{entry-condition-A}\\
& \varphi _x^ +  = \mathbb{B}' \varphi^{-*}, ~~ \varphi_x^-  = \mathbb{B}'^*  \varphi^{+*}, \label{entry-condition-B} \\
& \varphi^{\pm}_{t}=-4\varphi^{\pm}_{xxx}, \label{entry-condition-t}
\end{align}
\end{subequations}
in which $\mathbb{A}'=\t J_N[k_1]$ and $\mathbb{B}=J_N[k_1]$
with $k_1\in \mathbb{C}$.
Explicit forms of $\varphi^{\pm}$ are
\begin{equation}
\label{gen-sol-mkdv-31-new}
\varphi^+ =  \mathcal {A}\mathcal{Q}^{+}+\mathcal{B}\mathcal {Q}^{-},~~
\varphi^-= \mathcal{A}^{*} \mathcal{Q}^{+*}- \mathcal{B}^{*} \mathcal{Q}^{-*},
~~ \mathcal{A}, \mathcal{B} \in \widetilde{G}_N.
\end{equation}

\subsection{Rational solutions}\label{sec-4-3}

Recall the Remark \ref{rem-4-1} that indicates the LDES \eqref{cond-mKdV} fails in generating rational solutions
for the mKdV equation \eqref{mKdV} due to $|B|\neq 0$.
To derive rational solutions, we make use of the Galilean transformation (GT)
\begin{equation}
v(x,t)=v_{0}+V(X,t), ~~x=X+6{v_0}^{2}t, ~~v_0\in \mathbb{R},~ v_0\neq 0,
\label{GT}
\end{equation}
and consider the transformed equation (known also as the KdV-mKdV equation)
\begin{equation}
{V_t}+ 12{v_0}V{V_X} + 6{V^2}{V_X} + {V_{XXX}} = 0,
\label{kdv-mkdv-1}
\end{equation}
which admits rational solutions in Wronskian form. Once its rational solutions is obtained, one can reverse the GT \eqref{GT}
and get rational solutions of the mKdV equation \eqref{mKdV}.

By the transformation
\begin{equation}
V=i\,\biggl( \ln\frac{f^*}{f}\biggr)_X,
\label{trans-kdv-mkdv}
\end{equation}
\eqref{kdv-mkdv-1} is bilinearised as \cite{Wadati-JPSJ-1975}
\begin{subequations}
\begin{align}
& (D_t+D_{X}^{3})f^*\cdot f =0,\label{eq4a}
\\
& (D_{X}^{2}-2iv_{0}D_X)f^*\cdot f =0.\label{eq4b}
\end{align}
\label{bilinear-kdv-mkdv}
\end{subequations}

\begin{theorem}\label{thm-4-3-1}\cite{ZhaZSZ-RMP-2014}
The bilinear equation \eqref{bilinear-kdv-mkdv} admits Wronskian solution $f=|\h{N-1}|$,
where the elementary column vector $\varphi$ is determined by
\begin{subequations}
\label{cond1}
\begin{align}
& i\varphi _{X} =v_{0}\varphi +\mathbb{B} \varphi^*, \label{eq97}\\
& \varphi _{t} =-4\varphi_{XXX}.\label{eq96}
\end{align}
\end{subequations}
\end{theorem}

Note that solutions of the above system can be classified as in Sec.\ref{sec-4-2} by introducing
\begin{equation}\label{A-aux}
\varphi_{XX}=\mathbb{A}\varphi,
\end{equation}
where $\mathbb{A}=\mathbb{B} \mathbb{B}^{*}-v_0^{2}I_N$ and $I_N$ is the $N$th-order unit matrix.

$N$-soliton solution is obtained when taking
\begin{equation}
\mathbb{B}=\mathrm{diag}(-\sqrt{v_{0}^{2}+k_{1}^{2}},-\sqrt{v_{0}^{2}+k_{2}^{2}},\cdots,
-\sqrt{v_{0}^{2}+k_{N}^{2}}),
\end{equation}
and $\varphi$ satisfying \eqref{cond1} is composed by
\begin{equation}
\varphi_{j}=\sqrt{2v_{0}+2ik_{j}}\,e^{\eta_j}+\sqrt{2v_{0}-2ik_{j}}\,
e^{-\eta_j}, ~~\eta_j=k_{j}X-4k_{j}^{3}t, \label{phi-j}
\end{equation}
where $\{k_j\}$ are $N$ distinct real positive numbers.

To obtain rational solutions, we take $\mathbb{B}$ to be defined as \eqref{A-LTT} with
\begin{equation*}
a_{j}=\frac{-1}{(2j)!}\frac{\partial ^{2j} }{{\partial
k_{1}}^{2j}}\sqrt{v_{0}^{2}+k_{1}^{2}}\,\Bigr|_{k_{1}=0}.
\end{equation*}
In this case $\varphi_j$ is composed by
\begin{equation}
\varphi_{j+1}= \frac{1}{(2j)!}\frac{\partial ^{2j} }{{\partial
k_1}^{2j}}\varphi_{1}\,\bigr|_{k_1=0},~~(j=0,1,\cdots,N-1),
\end{equation}
where $\varphi_1$ is defined as in \eqref{phi-j}.

Write the corresponding Wronskian  as
$f=F_1+iF_2$, where $F_1=\mathrm{Re}[f],~ F_2=\mathrm{Im}[f]$.
From \eqref{trans-kdv-mkdv} solutions of the KdV-mKdV equation \eqref{kdv-mkdv-1} are expressed as
\begin{equation}
V(X,t)=-2 \,\frac{ F_{1,X}F_2-F_1F_{2,X}}{F_1^2+F_2^2},
\end{equation}
and for the mKdV equation \eqref{mKdV},
\begin{equation}
v(x,t)=v_{0}-\frac{2(F_{1,x}F_2-F_1F_{2,x})}{F_1^2+F_2^2},
\label{rat}
\end{equation}
where one needs to replace $X$ in $F_j$ with $~~ X=x-6v_{0}^{2}t$.
The simplest rational solution of the mKdV equation \eqref{mKdV} is
\begin{equation}
v =v_{0}-\frac{4v_{0}}{4v_{0}^{2}(x-6v_{0}^{2}t)^{2}+1}.
\label{1rs}
\end{equation}

\subsection{Notes}\label{sec-4-4}

Unlike the KdV equation, there is no positon-negaton-complexiton classification for the solutions of the mKdV equation
according to the eigenvalues of $\mathbb{B}$ in the LDES \eqref{mKdV-condition}.
In the simple-pole case, one can choose distinct $\{k_j\}$ from $\mathbb{C}/\sim$ to get solitons
where $\sim$ is defined by \eqref{sim}.

Analysis of direct scattering of the mKdV equation \eqref{mKdV} shows that the transmission coefficient $T(k)$ can have multiple poles \cite{WadO-JPSJ-82}.
Therefore multiple-pole solutions (limit solutions of solitons in this paper) are not singular.
Note that a typical characteristic of double pole solutions is that
at large time the two solitons asymptotically travel along logarithmic trajectories with a linear background.
This is also true for limit breathers \cite{ZhaZSZ-RMP-2014}.
Such a typical behavior was probably first found by Zakharov and Shabat
for double pole solution of the nonlinear Schr\"odinger (NLS) equation \cite{ZakS-JETP-1972}.
For more details of strict asymptotic analysis of this type of solutions one can refer to \cite{ZhaZSZ-RMP-2014,ZhoZZ-PLA-2009}.

Note specially that the mKdV equation \eqref{mKdV} serves as an integrable model that describes by its breathers the ultra-short pulse
propagation in a medium described by a two-level Hamiltonian \cite{LebM-PR-2013}.
Mixed solutions of the mKdV equation is potentially used to generate rogue waves \cite{SluP-PRL-2016}.
Such solutions correspond to the case that $\mathbb{B}$ in \eqref{mKdV-condition} or \eqref{cond1}
is block diagonal with diagonal and different Jordan cells, and the elementary column vector $\varphi$ is composed accordingly.

\section{The AKNS and reductions}\label{sec-5}

\subsection{The AKNS hierarchy and double Wronskian solutions}\label{sec-5-1}

A typical equation to admit double Wronskian solutions is the NLS equation \cite{Nimmo-PLA-1983-NLS},
which belongs to the well known AKNS hierarchy.
Let us recall some results of this hierarchy.

The AKNS hierarchy
\begin{equation}\label{akns-hierarchy}
 u_{t_n} =K_n     = \left(
                      \begin{array}{c}
                        K_{1,n} \\
                        K_{2,n}
                      \end{array}
                    \right)
                    =L^n \left(
                      \begin{array}{c}
                        -q \\
                        r
                      \end{array}
                    \right),
 ~~ u=\left(
   \begin{array}{c}
     q  \\
     r
   \end{array}
 \right),
 ~~n=1,2,\cdots,
\end{equation}
is derived from the AKNS spectral problem \cite{AKNS-PRL-1973}
\begin{equation}\label{akns-spectral}
 \Phi_x =\left( \begin{array}{cc}
                                  \lambda & q \\
                                  r & -\lambda
                                \end{array}
                              \right)  \Phi,~~\Phi=\left(
                                \begin{array}{c}
                                  \phi_1 \\
                                  \phi_2
                                \end{array}
                              \right),
\end{equation}
where $q$ and $r$ are functions of $(x,t)$, $\lambda$ is the spectral parameter, and the recursion operator is
\begin{equation*}
 L= \left(
            \begin{array}{ll}
              -\partial_x+2q \partial_x^{-1}r & 2q \partial_x^{-1}q \\
              -2r \partial_x^{-1}r & \partial_x- 2r \partial_x^{-1}q
            \end{array}
          \right).
\end{equation*}
By imposing
\begin{equation}\label{nls-hie-reduction1}
 r(x,t)= \delta q^*(\sigma x,t),~~\delta,\sigma=\pm 1
\end{equation}
on the even-indexed members of the AKNS hierarchy
\begin{equation}
iu_{t_{2l}}=-K_{2l},~~ l=1,2,\cdots,
\label{akns-hie-even}
\end{equation}
where we have replaced $t_{2l}$ by $it_{2l}$, one gets the NLS hierarchy
\begin{equation}
i q_{t_{2l}}= -K_{1,2l}|_{\eqref{nls-hie-reduction1}},~~ l=1,2,\cdots.
\label{nls-hie1}
\end{equation}
Note that  $(\sigma,\delta)=(1,\mp 1)$ yields the classical focusing and defocusing NLS hierarchy,
while  $\sigma=-1$ yields the so-called nonlocal NLS hierarchy (cf.\cite{AblM-PRL-2013}).
There are more reductions (cf.\cite{CheDLZ-SPAM-2018}), while here let us only focus on the
NLS hierarchy as an instructive example on double Wronskians and reductions.

After rewriting the AKNS hierarchy \eqref{akns-hierarchy} as its recursive form
\begin{equation}\label{akns-hierarchy-2}
 \left(
   \begin{array}{c}
     q_{t_{n+1}} \\
     r_{t_{n+1}} \\
   \end{array}
 \right)= L \left(
                      \begin{array}{c}
                        q_{t_{n}} \\
                        r_{t_{n}} \\
                      \end{array}
                    \right),~~n=1,2,\cdots
\end{equation}
and introducing  transformation
\begin{equation}\label{akns-transformation}
 q=\frac{h}{f},~~r=-\frac{g}{f},
\end{equation}
\eqref{akns-hierarchy-2} is bilinearised as (with $t_1=x$) \cite{Newell-book-1985}
\begin{subequations}\label{akns-bilibear}
\begin{align}
  &(D_{t_{n+1}}-D_{x}D_{t_n})g\cdot f =0, \label{akns-bilibear-1}  \\
  &(D_{t_{n+1}}-D_{x}D_{t_n})f\cdot h =0, \label{akns-bilibear-2}  \\
  &D^2_{x}f\cdot f =2gh.\label{akns-bilibear-3}
\end{align}
\end{subequations}

\begin{theorem}\label{thm-5-1}\cite{Liu-JPSJ-1990,YinSCC-CTP-2008}
The bilinear AKNS hierarchy \eqref{akns-bilibear} allows us the following double Wronskian solutions,
\begin{subequations}\label{akns-sol-fgh}
 \begin{align}
 & f= |\h{N-1};\h{M-1}|, \label{akns-f}  \\
 & g= 2^{N-M+1} |\h{N};\h{M-2}|, \label{akns-g}\\
 & h= 2^{M-N+1} |\h{N-2};\h{M}|, \label{akns-h}
 \end{align}
\end{subequations}
where the elementary column vectors \eqref{phipsi-NM} are  defined as
 \begin{equation}\label{akns-var-psi-AC}
  \varphi=\exp{\biggl(\frac{1}{2}\sum^{\infty}_{j=1} A^j t_j\biggr)}C^+, ~~\psi=\exp{\biggl( -\frac{1}{2}\sum^{\infty}_{j=1} A^j t_j\biggr)}C^-,
 \end{equation}
in which $A\in \mathbb{C}_{(N+M)\times (N+M)}$ is an arbitrary constant matrix and
\[C^{\pm}=(c_1^{\pm}, c_2^{\pm},\cdots, c_{N+M}^{\pm})^T,~~ c^{\pm}_i\in \mathbb{C}.\]
\end{theorem}

The proof is similar to the single Wronskian case as of Theorem \ref{thm-3-1-1}.
One can also refer to \cite{YinSCC-CTP-2008} for more details.

Note that \eqref{akns-var-psi-AC} provides solutions for the whole AKNS hierarchy \eqref{akns-hierarchy-2}
as well as any special equation $u_{t_n}=K_n$ in that hierarchy due to the theory of symmetries.
In the later case, the terms $e^{\pm \frac{1}{2} \sum_{j\neq n} A^{j} t_{j}}$ in \eqref{akns-var-psi-AC} are absorbed into $C^{\pm}$.
For the  hierarchy \eqref{akns-hie-even}, $\varphi$ and $\psi$ are taken as (with $t_1=x$)
\begin{equation}\label{akns-var-psi-AC-even}
\varphi=\exp{\biggl(\frac{1}{2}Ax+\frac{i}{2}\sum^{\infty}_{j=1} A^{2j} t_{2j}\biggr)}C^+, ~~
\psi=\exp{\biggl( -\frac{1}{2}Ax-\frac{i}{2}\sum^{\infty}_{j=1} A^{2j} t_{2j}\biggr)}C^-.
\end{equation}

\subsection{Reductions of double Wronskians}\label{sec-5-2}

As we have seen \eqref{akns-var-psi-AC} provides solutions through double Wronskians $f, g, h$ for the unreduced hierarchy \eqref{akns-hierarchy}.
In the following we present a simple reduction procedure that enables us to obtain double Wronskian solutions
for the reduced hierarchies. Let us take \eqref{akns-hie-even} and  \eqref{nls-hie1} as an example.
Note that $ \sigma, \delta =\pm 1$.

\begin{theorem}\label{thm-5-2}

Consider double Wroskians \eqref{akns-sol-fgh} with $\varphi$ and $\psi$ defined by \eqref{akns-var-psi-AC-even},
which provides solutions to the unreduced hierarchy \eqref{akns-hie-even}.
Impose the following constraints on $\varphi$ and $\psi$:
taking $M=N$, and
\begin{equation}
C^-=TC^{+*},
\label{const-NLS-C}
\end{equation}
where the $2N\times 2N$ matrix $T$ obeys
\begin{subequations}\label{const-NLS}
\begin{align}
& AT+\sigma TA^*=0,\label{const-NLS-TA}\\
& TT^*=\delta \sigma I_{2N}.\label{const-NLS-TT}
\end{align}
\end{subequations}
Then the NLS  hierarchy \eqref{nls-hie1} has solution
\begin{equation}\label{NLS-q-sol}
 q(x,t)=2 \frac{|\h{N-2}; \h{N}|}{|\widehat{N-1}; \widehat{N-1}|},
\end{equation}
where
\begin{equation}
\psi(x,t)=T \varphi^*(\sigma x,t).
\label{const-NLS-psi}
\end{equation}
\end{theorem}

\begin{proof}
For convenience we introduce notation
\begin{equation}
  \widehat{\varphi}^{(N)}(ax)_{[bx]} = \bigl(\varphi(ax),\partial_{bx}\varphi(ax),\partial_{bx}^2\varphi(ax),\cdots,\partial_{bx}^{N}\varphi(ax)\bigr),
  ~~a,b=\pm 1.
  \label{phipsi-ab}
\end{equation}
First, we show that under \eqref{const-NLS}  we have
\begin{align*}
\psi(\sigma x,t)=&\exp\biggl(-\frac{1}{2} \sigma Ax- \frac{i}{2}\sum^{\infty}_{j=1} A^{2j} t_{2j}\biggr)C^{-} \\
       =& \exp\biggl( \frac{1}{2} T A^*T^{-1}x-\frac{i}{2}\sum^{\infty}_{j=1} (TA^*T^{-1})^{2j} t_{2j}\biggr) TC^{+*}\\
       =& T \exp\biggl( \frac{1}{2}A^*x-\frac{i}{2} \sum^{\infty}_{j=1} A^{*2j} t_{2j}\biggr) C^{+*}  \\
       =& T \varphi^*(x,t),
\end{align*}
which gives \eqref{const-NLS-psi}.
Next, with the notation \eqref{phipsi-ab} and assumption \eqref{const-NLS-TT}, we have
\begin{align*}
f^*(\sigma x,t)&= |\widehat{\varphi}^{*(N-1)}(\sigma x)_{[\sigma x]}; T^* \widehat{\varphi}^{(N-1)}(\sigma^2 x)_{[\sigma x]}|    \\
& = |T^*|(\sigma\delta)^{N}|T \widehat{\varphi}^{*(N-1)}(\sigma x)_{[\sigma x]}; \widehat{\varphi}^{(N-1)}(x)_{[\sigma x]}|    \\
& = |T^*|(\sigma\delta)^{N}(-1)^{N}| \widehat{\varphi}^{(N-1)}(x)_{[x]}; T \widehat{\varphi}^{*(N-1)}(\sigma x)_{[x]}|   \\
& = |T^*|(\sigma\delta)^{N}(-1)^{N} f(x,t).
\end{align*}
In a similar manner,
\[g^*(\sigma x, t)=|T^*|\delta^{N+1}\sigma^{N}(-1)^{N-1} h(x,t).\]
Thus, we immediately reach
\[r^*(\sigma x,t)=- \frac{g^*(\sigma x, t)}{f^*(\sigma x, t)}=\frac{\delta h(x,t)}{f(x,t)}=\delta q(x,t),\]
i.e. the reduction \eqref{nls-hie-reduction1} for the NLS hierarchy.
\end{proof}

As for solutions $T$ and $A$  of \eqref{const-NLS}, if we assume they are block matrices of the form
\begin{equation}\label{TA-block}
 T=\left(
     \begin{array}{cc}
       T_1 & T_2 \\
       T_3 & T_4 \\
     \end{array}
   \right),~~A=\left(
                 \begin{array}{cc}
                   K_1 & \mathbf{0} \\
                   \mathbf{0} & K_4 \\
                 \end{array}
               \right),
\end{equation}
where $T_i$ and $K_i$ are $N\times N$ matrices,
then, solutions to \eqref{const-NLS} are given in Table \ref{tab-5-1}
where $\mathbf{K}_{N}\in \mathbb{C}_{N\times N}$.
\begin{table}[H]
\begin{center}
\begin{tabular}{|c|c|c|}
\hline
   $(\sigma, \delta)$    & $T$ &  $A$    \\
\hline
   $(1,-1)$              & $T_1=T_4=\mathbf{0}$, $T_3=-T_2 =\mathbf{I}_{N}$ & $ K_1=-K^*_4=\mathbf{K}_{N}$ \\
\hline
   $(1,1)$              & $T_1=T_4=\mathbf{0}$, $T_3=T_2 =\mathbf{I}_{N}$ & $ K_1=-K^*_4=\mathbf{K}_{N}$ \\
\hline
   $(-1,-1)$              & $T_1=T_4=\mathbf{0}$, $T_3=T_2 =\mathbf{I}_{N}$ & $ K_1=K^*_4=\mathbf{K}_{N}$ \\
\hline
   $(-1,1)$              & $T_1=T_4=\mathbf{0}$, $T_3=-T_2 =\mathbf{I}_{N}$ & $ K_1=K^*_4=\mathbf{K}_{N}$ \\
\hline
\end{tabular}
\caption{$T$ and $A$ for the NLS hierarchy}
\label{tab-5-1}
\end{center}
\end{table}

When $\mathbf{K}_{N}=\mathrm{Diag}(k_{1}, k_{2},  \cdots,  k_{N})$,
we have
\begin{align*}
\varphi= &\bigl(c_1 e^{\theta(k_1)}, c_2 e^{\theta(k_2)},\cdots,c_{N} e^{\theta(k_{N})},
            d_1 e^{\theta(-\sigma k_1^*)}, d_2 e^{\theta(-\sigma k_2^*)},\cdots,d_{N} e^{\theta(-\sigma k_{N}^*)}\bigr)^T,
\end{align*}
where
\begin{equation}\label{nls-theta}
\theta(k_l)=\frac{1}{2}k_l x+\frac{i}{2}\sum^{\infty}_{j=1}k_l^{2j}t_{2j}.
\end{equation}
When  $\mathbf{K}_{N}$ is the  Jordan matrix  $J_{N}(k)$ defined as in \eqref{mathbb-B-21},
we have
\begin{align*}
 \varphi= & \biggl( c e^{\theta(k)}, \frac{\partial_{k}}{1!}(c e^{\theta(k)}),\cdots, \frac{\partial^{N-1}_{k}}{(N-1)!} (c e^{\theta(k)}), \\
 & ~~~~~d e^{\theta(-\sigma k^*)}, \frac{\partial_{k^*}}{1!}(d e^{\theta(-\sigma k^*)}),\cdots,
 \frac{\partial^{N-1}_{k^*}}{(N-1)!}(d e^{\theta(-\sigma k^*)})\biggr)^T.
\end{align*}

Note that one more solution for the case $(\sigma,\delta)=(-1,\pm 1)$ is
\begin{align}
  T_1=-T_4=\sqrt{-\delta}\mathbf{I}_{N},~ T_2=T_3 =\mathbf{0}_{N},~~ K_1=\mathbf{K}_{N},~~K_4=-\mathbf{H}_{N},
  \label{TA-2}
\end{align}
where $\mathbf{K}_{N}, \mathbf{H}_{N}\in \mathbb{R}_{N\times N}$.
When
\[\mathbf{K}_{N}=\mathrm{Diag}(k_{1}, k_{2},\cdots,k_{N}),~~\mathbf{H}_{N}=\mathrm{Diag}(h_{1}, h_{2},\cdots,h_{N}),\]
we have
\begin{equation*}
\varphi= \Bigl(c_1 e^{\theta(k_1)}, c_2 e^{\theta(k_2)},\cdots,c_{N} e^{\theta(k_{N})},
 d_1 e^{\theta(-h_1)},  d_2 e^{\theta(-h_2)},\cdots, d_{N} e^{\theta(-h_{N})}\Bigr)^T,
\end{equation*}
and when $\mathbf{K}_{N}=J_{N}[k],~\mathbf{H}_{N}=J_{N}[h]$ as defined in \eqref{mathbb-B-21}, we have
\begin{align*}
 \varphi= & \biggl(c e^{\theta(k)}, \frac{\partial_{k}}{1!}(c e^{\theta(k)}),\cdots, \frac{\partial^{N-1}_{k}}{(N-1)!} (c e^{\theta(k)}),\\
 & ~~~~~~ d e^{\theta(-h)}, \frac{\partial_{h}}{1!}(d e^{\theta(-h)}),\cdots,\frac{\partial^{N-1}_{h}}{(N-1)!}(d e^{\theta(-h)})\biggr)^T,
\end{align*}
where $\theta(k)$ is defined in \eqref{nls-theta}.

\subsection{Notes}\label{sec-5-3}

The reduction technique we presented in this section is first introduced in \cite{CheZ-AML-2018}.
The technique is also valid for all one-field reductions of the AKNS hierarchy (cf.\cite{CheDLZ-SPAM-2018})
as well as for one-field reductions of other systems that allows us double Wronskian solutions.

It is known that the matrix $A$ in \eqref{akns-var-psi-AC} and its similar forms lead to same solutions for the AKNS hierarchy \eqref{akns-hierarchy},
and the eigenvalues of $A$ correspond to discrete spectrum of the AKNS
spectral problem  \eqref{akns-spectral}.
Therefore the structures of $A$ in Table \ref{tab-5-1}
indicate how the distribution of discrete spectrum changes with different reductions.

\section{Discrete case: the lpKdV equation}\label{sec-6}

\subsection{The lpKdV equation and discrete stuff}\label{sec-6-1}

There has been a surge of interest in discrete integrable systems in the last two decades (cf.\cite{HJN-book-2016} and the references therein).
Let us get familiar with some notations in discrete.
Suppose $u(n,m)$ to be a function defined on $\mathbb{Z}\times \mathbb{Z}$ where $n$ and $m$ are two discrete independent variables.
The basic operation on  $u(n,m)$ is a shift in stead of differentiation. In this section we employ notations
\begin{equation}
u\doteq u(n,m),~~ \widetilde{u}\doteq u(n+1,m),~~\widehat{u}\doteq u(n,m+1),~~\widehat{\widetilde{u}}\doteq u(n+1,m+1).
\label{nota-dis}
\end{equation}
A discrete Hirota's bilinear equation  is written as (cf.\cite{HieZ-JPA-2009})
  \begin{equation}
    \label{eq:HB}
    \sum_j\, c_j\, f_{j}(n+\nu_{j}^+,m+\mu_{j}^+)\,g_{j}(n+\nu_{j}^-,m+\mu_{j}^-)=0,
  \end{equation}
where $\nu_{i}^++\nu_{i}^-=\nu_{k}^++\nu_{k}^-,  \mu_{i}^++\mu_{i}^-=\mu_{k}^++\mu_{k}^-, \forall i,k$.
Casoratian is a discrete version of Wronskian.
Let
\begin{equation}\label{C-entry-vec}
\varphi(n,m,l)=(\varphi_1(n,m,l),\varphi_2(n,m,l),\cdots,\varphi_{N}(n,m,l))^T,
\end{equation}
where $\varphi_i$ are functions of $(n,m,l)$ defined on $\mathbb{Z}\times \mathbb{Z}\times \mathbb{Z}$.
A Casoratian w.r.t. $l$ together with its compact expression is
\begin{align}
  C(\varphi)=&|\varphi(n,m,0),\varphi(n,m,1),\cdots,\varphi(n,m,N-1)| \nonumber\\
  = & |\varphi(0),\varphi(1),\cdots,\varphi(N-1)|
  = |0,1,\cdots,N-1| =  |\h{N-1}|. \label{Caso-H1-1}
\end{align}
Similarly, we have $|\h{N-2},N|=|0,1,\cdots,N-2,N|$.

The lpKdV equation with the notations in \eqref{nota-dis} is written as
\begin{equation}
(u-\th{u})(\t{u}-\h{u})=b^2-a^2,
\label{lpKdV}
\end{equation}
which is a discrete version of the potential KdV equation, where $a$ and $b$ are respectively the spacing parameters of $n$ and $m$ directions.
By the transformation
\begin{equation}
  u = a n + b m + c_0  - \frac{g}{f},
\label{trans-lpkdv}
\end{equation}
where $c_0$ is a constant, \eqref{lpKdV} is bilinearised as \cite{HieZ-JPA-2009}
\begin{subequations}  \label{eq:bil-lpkdv}
\begin{align}
& \mathcal{H}_1=\h g\t f-\t g\h f+(a-b)(\h f\t f- f\th f)=0, \label{eq:bil-lpkdv-a}\\
& \mathcal{H}_2=g\th f-\th g f+(a+b)(f\th f-\h f\t f)=0. \label{eq:bil-lpkdv-b}
\end{align}
\end{subequations}

\subsection{Casoratian solutions}\label{sec-6-2}

\begin{theorem}\label{thm-6-1}
The bilinear lpKdV equation \eqref{eq:bil-lpkdv} admits Casoratian solutions
\begin{equation}
f(\varphi)=|\h{N-1}|,\,g(\varphi)=|\h{N-2},N|,
\label{cas-f/g-1}
\end{equation}
where the elementary column vector $\varphi(n,m,l)$ satisfies
\begin{subequations}
\label{caso-cond}
\begin{align}
& \t\varphi-\b \varphi =(a-c) \varphi,\label{caso-cond-nl}\\
& \h\varphi-\b \varphi =(b-c) \varphi,\label{caso-cond-ml}
\end{align}
and there exists an auxiliary vector $\psi$ and an invertible matrix
$\Gamma=\Gamma(m)$ such that
\begin{align}
& \varphi = \Gamma \psi,\label{caso-cond-GPP}\\
& \psi + \b \psi= (b+c) \h\psi. \label{caso-cond-psi-ml}
\end{align}
\end{subequations}
Here $c$ is a constant and $\b f(n,m,l)\doteq f(n,m,l+1)$.
Note that $\Gamma$ is  independent of $n,l$,
and it then follows from \eqref{caso-cond-ml} and \eqref{caso-cond-GPP} that
\begin{equation}
\t\psi-\b \psi =(a-c) \varphi.\label{caso-cond-psi-nl}
\end{equation}
\end{theorem}

\begin{proof}
Making use of evolution relations \eqref{caso-cond-nl} and \eqref{caso-cond-ml},
one can derive shift relations of $f$ and $g$ \cite{HieZ-JPA-2009},
\begin{subequations}
\begin{align}
& -(a-c)^{N-2}\dt f  =|\h{N-2},\dt \varphi(N-2)|,\label{f-dt}\\
& -(b-c)^{N-2}\dh f  =|\h{N-2},\dh \varphi(N-2)|,\\
& -(a-c)^{N-2}[\dt g+(a-c) \dt f]=|\h{N-3},N-1, \dt \varphi(N-2)|,\label{fg-dt}\\
& -(b-c)^{N-2}[\dh g+(b-c) \dh f]=|\h{N-3},N-1, \dh \varphi(N-2)|.
\end{align}
\end{subequations}
Then, for down-tilde-hat-shifted \eqref{eq:bil-lpkdv-a}, we find
\begin{align*}
  & [(a-c)(b-c)]^{N-2}[-(a-b)f\dth f + \dh f (\dt g+(a-c)\dt f) -\dt f(\dh g + (b-c)\dh f)]\\
  =& -|\h{N-1}|  |\h{N-3},\dh{\varphi}(N-2),\dt{\varphi}(N-2)|\\
   & +|\h{N-2},\dh{\varphi}(N-2)| |\h{N-3},N-1,\dt{\varphi}(N-2)| \\
   & -|\h{N-2},\dt{\varphi}(N-2)|  |\h{N-3},N-1,\dh{\varphi}(N-2)|
\end{align*}
which vanishes in light of \eqref{plu-r-1} where   $M=(\h{N-3})$,
and $\mathbf{a}$, $\mathbf{b}$, $\mathbf{c}$ and $\mathbf{d}$
take $\varphi(N-2)$, $\varphi(N-1)$, $\dh{\varphi}(N-2)$ and $\dt{\varphi}(N-2)$, respectively.

To prove \eqref{eq:bil-lpkdv-b}, we make use of the auxiliary vector $\psi$ and
consider Casoratian $f(\psi)$, which evolutes in $m$ direction as
\begin{equation*}
(b+c)^{N-2}\h f(\psi)=|\psi(0),\psi(1),\cdots, \psi(N-2),\h\psi(N-2)|.
\end{equation*}
Then, noting that $\varphi=\Gamma \psi$ and $f(\varphi)=|\Gamma|f(\psi)$, we recover $\h f(\varphi)$ as
\begin{align*}
(b+c)^{N-2} \h f(\varphi) =&(b+c)^{N-2}|\h \Gamma| \h f(\psi)\\
= &|\h \Gamma||\Gamma^{-1}||\varphi(0),\cdots,\varphi(N-2), \c E \varphi(N-2)|,
\end{align*}
i.e.
\begin{equation}
(b+c)^{N-2} \h f(\varphi) = |\h \Gamma||\Gamma^{-1}||\h{N-2}, \c E \varphi(N-2)|,
\label{f-hat}
\end{equation}
where $\c E \varphi(l)=\Gamma\h \Gamma^{-1} \h \varphi(l)$.
With the help of $f(\psi)$ and $g(\psi)$, we also get
\begin{equation}
  (a+b) [(a-c)(b+c)]^{N-2} \dt{\h f}(\varphi)=|\h \Gamma||\Gamma^{-1}| |\h{N-3},\dt \varphi(N-2),\c E \varphi(N-2)|
  \label{f-hat-dt}
\end{equation}
and
\begin{equation}
(b+c)^{N-2}[\h g(\varphi)-(b+c) \h f(\varphi)] =|\h \Gamma||\Gamma^{-1}||\h{N-3},N-1,\c E \varphi(N-2)|.
\label{fg-h}
\end{equation}
Then, with formulas \eqref{f-dt}, \eqref{fg-dt}, \eqref{f-hat}, \eqref{f-hat-dt} and \eqref{fg-h},
one can verify \eqref{eq:bil-lpkdv-b} in its down-tilde-shifted version.

Thus we complete the proof.
\end{proof}

Explicit $\varphi$ and $\psi$ that satisfy \eqref{caso-cond} can be given according to the canonical forms of $\Gamma$.
When
\begin{equation}\label{Gamma-dia}
\Gamma=\mathrm{Diag}(\gamma_{1},\gamma_{2},\cdots,\gamma_{N}),~~\gamma_{j}=(b^2-k_j^2)^m
\end{equation}
with distinct $k_j$, one can take
\begin{align}
& \varphi_i(n,m,l)=\varphi^+_i+\varphi^{-}_i,~~\varphi^{\pm}_i= \varrho_{i}^{\pm}(c\pm k_i)^l(a \pm k_i)^n(b \pm k_i )^m,\label{varp-gen}\\
& \psi_i(n,m,l)=\psi^+_i+\psi^{-}_i,~~
\psi^{\pm}_i= \varrho_{i}^{\pm}(c\pm k_i)^l(a \pm k_i)^n(b \mp k_i )^{-m},
\label{psi-gen}
\end{align}
where $\varrho_{i}^{\pm}\in \mathbb{C}$.
When $\Gamma$ is a LTTM defined as
\eqref{A-LTT} with $a_j=\frac{1}{j!}\partial^{j}_{k_1}\gamma_1$
where $\gamma_1$ is defined in \eqref{Gamma-dia}, one can take
\begin{subequations}
\begin{align}
& \varphi(m,n,l)=\mathcal{A}^+\mathcal{Q}^{+} +\mathcal{A}^{-}\mathcal{Q}^{-},~~ \mathcal{A}^{\pm}\in \t G_N,
\label{varp-jor-1}\\
& \psi(m,n,l)=\mathcal{B}^+\mathcal{P}^{+} +\mathcal{B}^{-}\mathcal{P}^{-},~~ \mathcal{B}^{\pm}\in \t G_N,
\label{psi-jor-1}
\end{align}
where
\begin{align}
& \mathcal{Q}^{\pm} =(Q^{\pm}_{0},Q^{\pm}_{1},\cdots, Q^{\pm}_{N-1})^T,~~
Q^{\pm}_{s} =\frac{1}{s!}\partial^{s}_{k_1}\varphi^{\pm}_1,
\label{Q-k1}\\
& \mathcal{P}^{\pm} =(P^{\pm}_{0},P^{\pm}_{1},\cdots, P^{\pm}_{N-1})^T,~~
P^{\pm}_{s} =\frac{1}{s!}\partial^{s}_{k_1}\psi^{\pm}_1.
\label{P-k1}
\end{align}
\end{subequations}
and
$\varphi^{\pm}_1$ and $\psi^{\pm}_1$ are defined in \eqref{varp-gen} and \eqref{psi-gen} respectively.
When $\Gamma$ is a LTTM defined as
\eqref{A-LTT} with $a_j=\frac{1}{(2j)!}\partial^{2j}_{k_1}\gamma_1|_{k_1=0}$
where $\gamma_1$ defined in \eqref{Gamma-dia}, one can take
\begin{subequations}
\label{varp-jor-rs}
\begin{equation}
\varphi(m,n,l)=\mathcal{A}^+\mathcal{R}^{+} +\mathcal{A}^{-}\mathcal{R}^{-},~~ \mathcal{A}^{\pm}\in \t G_N,
\end{equation}
where $\mathcal{R}^{\pm} =(\mathcal{R}^{\pm}_{0},\mathcal{R}^{\pm}_{1},\cdots, \mathcal{R}^{\pm}_{N-1})^T$ and
\begin{align}
& \mathcal{R}^{+}_{s} =\frac{1}{(2s)!}\partial^{2s}_{k_1}\varphi_1|_{k_1=0},~\mathrm{with}~\varrho^-_1=\varrho^+_1,\\
& \mathcal{R}^{-}_{s} =\frac{1}{(2s+1)!}\partial^{2s+1}_{k_1}\varphi_1|_{k_1=0},~\mathrm{with}~\varrho^-_1=-\varrho^+_1,
\end{align}
\end{subequations}
with $\varphi_1$ defined in \eqref{varp-gen}. This case yields rational solutions.

\subsection{Notes}\label{sec-6-3}

The system \eqref{caso-cond} is the ``LDES'' of the bilinear lpKdV equation,
in which $\psi$ and $\Gamma$ have been used as auxiliaries in order to implement Casoration verifications.
More formulas of shifted Casotarians can be found in \cite{HieZ-JPA-2009}
with the help of auxiliary vectors.
Examples that Theorem \ref{thm-2-3-1} plays its role in discrete case can be found in \cite{HieZ-JMP-2010,HieZ-SIGMA-2011}.
It can be proved that $\Gamma$ and its similar forms lead to same solutions to the lpKdV equation \eqref{lpKdV},
which means one can make use of  canonical forms of $\Gamma$ to derive and classify solutions,  including rational solutions.
For more results on rational solutions in Casoration form for fully discrete 2D integrable systems
one can refer to \cite{ShiZ-SIGMA-2011,ZhaZ-SIGMA-2017,ZhaZ-JNMP-2019}.

\eqref{varp-gen} is a discrete counterpart of the continuous exponential function (e.g. \eqref{entry-Case1-1}),
where $\pm k_i$ satisfy $x^2-k^2_i=0$.
In principle, for fully discrete 2D integrable systems, their dispersion relations are defined by the curve
\begin{equation}
P_M(x,k)=\sum^{M}_{i=1}a_i(x^i-k^i)=0,
\end{equation}
solutions of which  are used to define discrete exponential functions.
For more details one can refer to \cite{HieZ-SIGMA-2011,ZhaZN-SAPM-2012}.

Besides, the techniques and treatments used in Sec.\ref{sec-3} and \ref{sec-4} for Wronskians
are also valid to Casoratians. For example, the technique to verify BT with Wronskian solutions in Sec.\ref{sec-3-2}
has applied to Casotatians \cite{ZhoZZ-PLA-2009},
and the reduction procedure of double Wronskians described in Sec.\ref{sec-5-2}
has been also generalised to doubel Casoratians \cite{DenLZ-AMC-2018},
although \cite{DenLZ-AMC-2018} and \cite{ZhoZZ-PLA-2009} describe semidiscrete models.

\section{Conclusions}\label{sec-7}

We have reviewed Wronskian technique and solutions in Wronskian/Casoratian forms for continuous and discrete integrable systems.
In this context four instructive examples were employed.
By the KdV equation we showed standard verifying procedures of Wronskian solutions of the bilinear KdV equation and the bilinear BT.
It also servered as an example that displays the construction of  general solutions of the LDES and
furthermore the classification of solutions for the KdV equation
according to the canonical forms of the coefficient matrix $A$ in its LDES.
The second example is the mKdV equation which is special in many aspects.
Note that there is a complex conjugate operation in its LDES;
there is no solution classification as the KdV equation;
breathers in Wronskian form result from block diagonal $\mathbb{B}$  \eqref{B-breather};
solitons or breathers of double pole case  are not singular and travel asymptotically with  logarithmic curves in stead of straight lines;
rational solutions are obtained by means of GT \eqref{GT}.
The third example is the AKNS hierarchy, together with its double Wronskian solutions and reductions.
A reduction technique was showed to get double Wronskian solutions for the reduced hierarchy.
Finally, the lpKdV equation served as an example of fully discrete integrable systems
which have received significant progress in the recent two decades.
We described how to obtain shift formulas of Casoratians by introducing auxiliary vectors, so that Casoratians solutions of
discrete bilinear equations can be verified as in the continuous case.

More than thirty years passed since the Wronskian technique was proposed in 1983 \cite{FreN-PLA-1983}.
Almost no secret is left behind this technique.
As a matter of fact, solutions in Wronskian form have their own advantage
in presenting explicit multiple-pole solutions (including rational solutions)
and understanding relations between simple-pole and multiple-pole solutions
by taking limits w.r.t. poles $\{k_j\}$.
This benefits from  the regular structure of a Wronskian
that each row is governed by a single $k_j$.
Note that the IST \cite{AblS-1981}, Cauchy matrix approach \cite{NijAH-JPA-2009} and operator approach (see \cite{Sch-LAA-2010}
and the references therein) yield solutions in terms of the Cauchy matrix,
but the expression for multiple-pole solutions are not as simple as those in Wronskian case (cf.\cite{ZhaZ-SAPM-2013}),
and so far no rational solutions are presented in terms of  Cauchy matrix.

There are other compact expressions for solutions of integrable systems, such as Grammian and Pfaffian, mainly derived from bilinear methods.
For more details one may refer to \cite{Hir-book-2004}.

\section*{Acknowledgements}

The project is supported by the NSF of China (Nos.11631007 and 11875040).

\end{document}